\documentclass[sigconf,nonacm]{acmart}
\usepackage[utf8]{inputenc}
\usepackage{amsmath,amsthm}
\usepackage[lambda,landau,operators,probability,sets,logic,advantage,complexity,asymptotics,notions,mm,primitives,ff,adversary,keys,events,oracles]{cryptocode}
\usepackage{hyperref}
\usepackage{xspace}
\usepackage{bm}
\usepackage{dashbox}
\usepackage{seqsplit}
\usepackage{booktabs}
\usepackage{ae,aecompl}
\usepackage{caption,subcaption}
\usepackage{enumitem}

\usepackage{pifont}
\newcommand{\cmark}{\ding{51}}%
\newcommand{\xmark}{\ding{55}}%

\usepackage{tikz}
\usetikzlibrary{shapes,arrows,shadows,positioning,backgrounds}

\newcommand{\D}{\ensuremath{\mathcal{D}}\xspace}
\newcommand{\E}{\ensuremath{\mathcal{E}}\xspace}
\newcommand{\F}{\ensuremath{\mathcal{F}}\xspace}
\newcommand{\gdv}{\ensuremath{\mathcal{G}}\xspace}
\renewcommand{\H}{\ensuremath{\mathcal{H}}\xspace}

\newcommand{\M}{\ensuremath{\mathcal{M}}\xspace}
\newcommand{\N}{\ensuremath{\mathcal{N}}\xspace}
\renewcommand{\P}{\ensuremath{\mathcal{P}}\xspace}
\newcommand{\Q}{\ensuremath{\mathcal{Q}}\xspace}

\newcommand{\T}{\ensuremath{\mathcal{T}}\xspace}
\newcommand{\udv}{\ensuremath{\mathcal{U}}\xspace}

\newcommand{\W}{\ensuremath{\mathcal{W}}\xspace}
\newcommand{\X}{\ensuremath{\mathcal{X}}\xspace}
\newcommand{\Y}{\ensuremath{\mathcal{Y}}\xspace}
\newcommand{\Z}{\ensuremath{\mathcal{Z}}\xspace}

\newcommand{\OO}{\ensuremath{\mathbb{O}}\xspace}
\renewcommand{\SS}{\ensuremath{\mathbb{S}}\xspace}
\newcommand{\threshold}{\ensuremath{\kappa}\xspace}

\newcommand{\RO}{\ensuremath{\mathsf{RO}}\xspace}

\newcommand{\push}[1]{\ensuremath{\mathsf{push}(#1)}\xspace}
\newcommand{\mathcolorbox}[2]{\colorbox{#1}{$\displaystyle #2$}}
\newcommand{\goodcell}[1]{\cellcolor{white}{#1}}
\newcommand{\avgcell}[1]{\cellcolor{white}{#1}}
\newcommand{\badcell}[1]{\cellcolor{white}{#1}}

\newcommand{\skescheme}{\ensuremath{\mathsf{ske}}\xspace}
\newcommand{\skeschemeshort}{\ensuremath{\Sigma}\xspace}
\newcommand{\aux}{\ensuremath{\mathsf{aux}}\xspace}

\newcommand{\adss}[2]{\ensuremath{(#1,#2)\text{-secret-sharing}}\xspace}

\newcommand{\adssschemeshort}[2]{\ensuremath{\Pi}_{#1,#2}\xspace}
\newcommand{\share}[1]{\ensuremath{\mathsf{share}(#1)}\xspace}
\newcommand{\reveal}[1]{\ensuremath{\mathsf{recover}(#1)}\xspace}
\newcommand{\sharenoinput}{\ensuremath{\mathsf{share}}\xspace}
\newcommand{\revealnoinput}{\ensuremath{\mathsf{recover}}\xspace}

\newcommand{\mtag}{\ensuremath{\tau}\xspace}

\newcommand{\voprf}{VOPRF\xspace}
\newcommand{\voprfscheme}{\ensuremath{\mathsf{voprf}}\xspace}
\newcommand{\voprfschemeshort}{\ensuremath{\Gamma}\xspace}

\newcommand{\mpk}{\ensuremath{\mathsf{mpk}}\xspace}
\newcommand{\msk}{\ensuremath{\mathsf{msk}}\xspace}
\newcommand{\pparams}{\ensuremath{\mathsf{pp}}\xspace}
\renewcommand{\rq}{\ensuremath{\mathsf{rq}}\xspace}
\renewcommand{\state}{\ensuremath{\mathsf{st}}\xspace}
\newcommand{\rp}{\ensuremath{\mathsf{rp}}\xspace}

\newcommand{\derive}[1]{\ensuremath{\mathsf{derive}(#1)}\xspace}
\newcommand{\poprfsetup}[1]{\ensuremath{\mathsf{setup}(#1)}\xspace}
\newcommand{\mskgen}[1]{\ensuremath{\mathsf{keygen}(#1)}\xspace}
\newcommand{\req}[1]{\ensuremath{\mathsf{req}(#1)}\xspace}
\renewcommand{\eval}[1]{\ensuremath{\mathsf{eval}(#1)}\xspace}

\newcommand{\finalize}[1]{\ensuremath{\mathsf{finalize}(#1)}\xspace}

\newcommand{\derivenoinput}{\ensuremath{\mathsf{derive}}\xspace}

\newcommand{\ideal}[1]{\ensuremath{\mathsf{Ideal}(#1)}\xspace}
\newcommand{\real}[1]{\ensuremath{\mathsf{Real}(#1)}\xspace}

\newcommand{\inputs}{\ensuremath{\mathsf{inputs}}\xspace}
\newcommand{\stat}{\ensuremath{\stackrel{\text{\tiny s}}{\simeq}}\xspace}
\newcommand{\comp}{\ensuremath{\stackrel{\text{\tiny c}}{\simeq}}\xspace}
\newcommand{\chkequal}{\ensuremath{\stackrel{\text{\tiny ?}}{=}}\xspace}
\newcommand{\leakage}[1]{\ensuremath{\mathsf{L}(#1)}\xspace}
\newcommand{\leakagenoinput}{\ensuremath{\mathsf{L}}\xspace}
\newcommand{\wedgespace}{\ensuremath{\,\wedge\,}\xspace}

\newcommand{\game}[1]{\ensuremath{\gdv_{#1}}\xspace}

\newcommand{\point}[1]{\par\smallskip{\noindent\textbf{#1.}~}}

\newcommand{\tool}[1][]{\textsf{STAR}{#1}\xspace}
\newcommand{\toolfull}{\tool{}\xspace}
\newcommand{\toollite}{\tool{}\textsf{Lite}\xspace}
\newcommand{\titletext}{\texorpdfstring{\tool}{STAR}: \underline{S}ecret Sharing for Private \underline{T}hreshold \underline{A}ggregation \underline{R}eporting}
\newcommand{\repourl}{\footnote{\anonymous{\url{https://bitbucket.org/star-protocol/sta-rs/}}{\url{https://github.com/brave-experiments/sta-rs}}}\xspace}

\hypersetup{colorlinks=true,citecolor=blue,linkcolor=blue,linkbordercolor=red,pdfborderstyle={/S/U/W 1}}

\usepackage{xcolor}
\usepackage{colortbl}

\theoremstyle{plain}
\newtheorem{theorem}{Theorem}

\newtheorem{claim}{Claim}[section]
\newtheorem{remark}[theorem]{Remark}

\def\isanonymous{0}
\usepackage{ifthen}
\newcommand{\anonymous}[2]{%
\ifthenelse{\equal{\isanonymous}{1}}
{{#1}}%
{{#2}}%
}%

\def\isfullversion{1}
\usepackage{ifthen}
\newcommand{\includefull}[2]{%
\ifthenelse{\equal{\isfullversion}{1}}
{{#1}}%
{{#2}}%
}%

\definecolor{alexpink}{HTML}{FFA0A0}
\usepackage{todonotes}

\author{Alex Davidson}
\affiliation{%
  \institution{Brave Software}
  \country{}
}
\author{Peter Snyder}
\affiliation{%
    \institution{Brave Software}
    \country{}
}
\author{E. B. Quirk}
\affiliation{%
    \institution{Brave Software}
    \country{}
}
\author{Joseph Genereux}
\affiliation{%
    \institution{Brave Software}
    \country{}
}
\author{Benjamin Livshits}
\affiliation{%
    \institution{Imperial College London}
    \country{}
}
\author{Hamed Haddadi}
\affiliation{%
    \institution{Brave Software}
    \country{}
}
\affiliation{%
    \institution{Imperial College London}
    \country{}
}

\title{\titletext}
\date{}

\begin{document}
\sloppy

% CCS stuff
\begin{abstract}
Threshold aggregation reporting systems promise a practical, privacy-preserving solution for developers to learn how their applications are used ``\emph{in-the-wild}''. Unfortunately, proposed systems to date prove impractical for wide scale adoption, suffering from a combination of requiring: \textit{i)} prohibitive trust assumptions; \textit{ii)} high computation costs; or \textit{iii)} massive user bases. As a result, adoption of truly-private approaches has been limited to only a small number of enormous (and enormously costly) projects.

In this work, we improve the state of private data collection by proposing \tool, a highly efficient, easily deployable system for providing cryptographically-enforced \(\threshold\)-anonymity protections on user data collection. The \tool protocol is easy to implement and cheap to run, all while providing privacy properties similar to, or exceeding the current state-of-the-art. Measurements of our open-source implementation of \tool{} find that it is \(1773\times\) quicker, requires \(62.4\times\) less communication, and is \(24\times\) cheaper to run than the existing state-of-the-art.
\end{abstract}

% \keywords{threshold aggregation; private analytics}
%     <ccs2012>
%     <concept>
%         <concept_id>10002978.10002991.10002995</concept_id>
%         <concept_desc>Security and privacy~Privacy-preserving protocols</concept_desc>
%         <concept_significance>500</concept_significance>
%     </concept>
%     </ccs2012>
% \ccsdesc[500]{Security and privacy~Privacy-preserving protocols}
    
\maketitle
%% arxiv
\thispagestyle{empty}
    
\section{Introduction}
\label{sec:intro}
%!TEX root=../main.tex

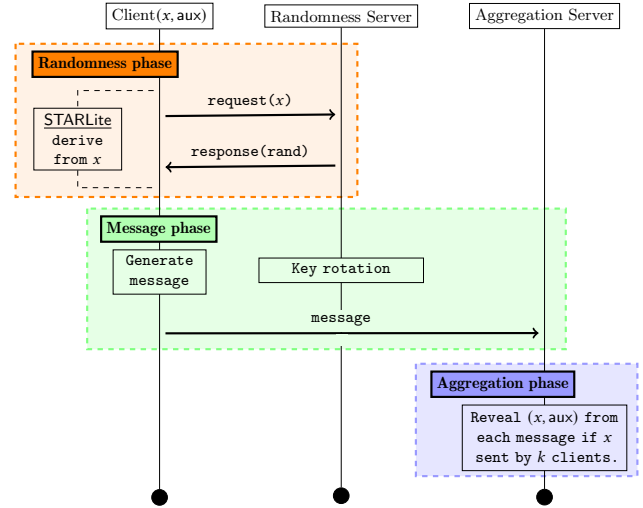
\begin{figure}[t]
    \centering
    \scalebox{0.68}{%
    \begin{tikzpicture}
        \node[rectangle,draw=black] (C) {Client\((x,\aux)\)};
        \node[rectangle,draw=black,right=of C] (R) {Randomness Server};
        \node[rectangle,draw=black,right=of R] (A) {Aggregation Server};
        \node[circle,draw=black,fill=black,below=of C,yshift=-8cm] (endC) {};
        \node[circle,draw=black,fill=black,below=of R,yshift=-8cm] (endR) {};
        \node[circle,draw=black,fill=black,below=of A,yshift=-8cm] (endA) {};
        \draw[black] (C) -- 
            node[yshift=2.8cm] (c1) {} 
            node[yshift=3.3cm] (c1above) {} 
            node[xshift=-1.6cm,yshift=2.35cm,draw=black,fill=orange!10,text width=1.5cm,align=center] (c1lite) {\underline{\toollite} \texttt{derive from} \(x\)} 
            node[yshift=1.8cm] (c2) {} 
            node[yshift=1.4cm] (c2below) {} 
            node[draw=black,fill=green!10,yshift=-0.25cm,text width=1.6cm,align=center] (c4) {\texttt{Generate message}}
            node[yshift=-1.45cm] (c5) {}
        (endC);
        \draw[black] (R) -- 
            node[yshift=2.8cm] (r1) {} 
            node[yshift=1.8cm] (r2) {}
            node[draw=black,fill=green!10,text width=3cm,align=center,yshift=-0.25cm] (r3) {\texttt{Key rotation}}
            node[yshift=-1.45cm] (r4) {}
            (endR);
        \draw[black] (A) -- 
            node[yshift=-1.45cm] (a1) {}
            node[yshift=-3.5cm,draw=black,fill=blue!10,text width=3cm,align=center] (a2) {\texttt{Reveal \((x,\aux)\) from each message if \(x\) sent by \(k\) clients.}}
        (endA);
        \draw[->,very thick] (c1) -- node[above] {\texttt{request(\(x\))}} (r1);
        \draw[->,very thick] (r2) -- node[above] {\texttt{response(rand)}} (c2);
        \draw[->,very thick] (c5) -- node[above,yshift=0.012cm,xshift=-0.21cm,fill=green!10] {\texttt{message}} node[below,yshift=-0.006cm,xshift=-0.2cm,fill=green!10,minimum height=0.1cm,minimum width=0.1cm] {} (a1);
        \draw[dashed] (c1above) -| (c1lite);
        \draw[dashed] (c1lite) |- (c2below);
        \begin{scope}[on background layer]
            \draw[orange,fill=orange!10,very thick,dashed] ($(c1.north west)+(-2.7,1.3)$)  rectangle ($(r2.south east)+(0.2,-0.5)$);
            \draw[green!50,fill=green!10,very thick,dashed] ($(c4.north west)+(-0.5,0.8)$)  rectangle ($(a1.south east)+(0.3,-0.2)$);
            \draw[blue!40,fill=blue!10,very thick,dashed] ($(a2.north west)+(-0.9,0.8)$)  rectangle ($(a2.south east)+(0.2,-0.1)$);
        \end{scope}
        \node[fill=orange,very thick,draw=black] at ($(c1.north west)+(-0.95,0.9)$) {\textbf{Randomness phase}};
        \node[fill=green!30,very thick,draw=black] at ($(c4.north west)+(0.9,0.4)$) {\textbf{Message phase}};
        \node[fill=blue!40,very thick,draw=black] at ($(a2.north west)+(0.8,0.4)$) {\textbf{Aggregation phase}};
    \end{tikzpicture}} 
    \caption{ General \tool architecture. In the \textbf{Randomness
    phase}, clients sample randomness from a dedicated server. In the \textbf{Message phase}, clients generate their messages to send to the aggregation server. The
    aggregation server learns those measurements sent by \threshold
    clients in the \textbf{Aggregation phase}. Client
    randomness can be sampled locally, if the measurement distribution
    is sufficiently entropic (\toollite, Section~\ref{sec:starlite-dists}).}
    \label{fig:generalarch}
\end{figure}

Application developers often need to learn how their product is used, and in which environments their software runs. Such information helps developers debug errors, address security issues, and optimize implementations.

However, collecting such information puts user privacy at risk. Among other concerns, collecting user data, even de-identified, may allow a data collector to profile a user or link records, revealing increasingly rich information about users over time. Naive data collection can harm user privacy either in ways unintended by the developer or unexpected by the user, or both.

A common approach for protecting user privacy when collecting client measurement data is to only learn those measurements that are sent by \(\threshold\) clients (sometimes called \(\threshold\)-heavy-hitters). In these systems, the central server only learns the measurement if there are at least \(\threshold-1\) other clients that provide it as well. This approach prevents the data collector from learning uniquely identifying (or uniquely co-occurring patterns of) values, with the broader goal of preventing the identification of any individuals in aggregate dataset. Such guarantees are strongly related to the privacy notion of \(\threshold\)-anonymity~\cite{IJUFKBS:Sweeney02}. We refer to systems that can provide such guarantees as \textbf{threshold aggregation systems}.

Designers of threshold aggregation systems face a challenging dichotomy though: how to allow a server to determine if it has collected \(\threshold\) identical records, without: i) the server first seeing the underlying measurement value; and ii) in a manner that protects the user against a malicious (or generally untrusted) server.

Many systems have been proposed to try and square this circle~\cite{prochlo,NSDI:CorBon17,SP:BBCGI21,ACMTA:BNS19,PMLR:ZKMSL20,JMLR:BNST20,STOC:BasSmi15,CCS:QYYKXR16,C:KisSon05,Chaum81,CCS:Neff01,AsiaCCS:BlaAgu12,CCS:ErlPihKor14}. However, all such systems to date have properties that make them impractical for most developers and telemetry systems. More specifically, all systems to date have at least one of the following undesirable properties:
\begin{itemize}
    \itemsep0em
    \parskip0em
    \parsep0em
    \item expensive server-side aggregation~\cite{C:KisSon05,AsiaCCS:BlaAgu12};
    \item non-collusion assumptions for servers that communicate with each other~\cite{prochlo,SP:BBCGI21};
    \item interactive communication between clients~\cite{C:KisSon05,Chaum81,CCS:Neff01};
    \item trusted third parties or hardware~\cite{prochlo};
    \item difficult to apply for cases where \(\threshold > 1\)~\cite{Chaum81,CCS:Neff01};
    \item require noise injection, and so require large user bases and/or entail utility loss~\cite{ACMTA:BNS19,PMLR:ZKMSL20,JMLR:BNST20,STOC:BasSmi15,CCS:QYYKXR16,CCS:ErlPihKor14};
    \item restricted to numeric data types~\cite{NSDI:CorBon17,EPRINT:AGJO21};
    \item unbounded worst case leakage~\cite{ACMTA:BNS19,PMLR:ZKMSL20,JMLR:BNST20,STOC:BasSmi15,CCS:QYYKXR16}.
\end{itemize}

\subsection{The \texorpdfstring{\emph{\tool}}{STAR} approach}

In response to these issues in current threshold aggregation systems, we propose \tool{}; a practical, private threshold aggregation system that prioritizes \emph{i)} efficiency (so that it can be deployed at extremely low cost), \emph{ii)} limited trust assumptions (so that the trust requirements can be achieved by a wider range projects), and \emph{iii)} simple, well-established cryptography (so that systems can be implemented and audited by a wider range of developers).

Further, \tool{} provides capabilities existing threshold aggregation systems lack, allowing \tool{} to solve use cases unaddressed by current state-of-the-art. Specifically, \tool{} allows developers to attach arbitrary (but still threshold-protected) data to client messages.

\point{Overall idea} Figure~\ref{fig:generalarch} presents an overview of the \tool{} approach. Each client constructs a ciphertext by encrypting their measurement (and any auxiliary data) using an encryption key derived deterministically from i) any randomness present in the client measurement and ii) additional randomness provided by a ``randomness server''.  This randomness server never learns client values or inputs.

The client then sends: i) the ciphertext; ii) a \(\threshold\)-out-of-\(N\) secret share of the randomness used to derive the encryption key; and iii) a deterministic tag informing the server which shares to combine. The aggregation server groups reports with the same tag, and recovers the encryption keys from those subsets of size \(\geq \threshold\). Thus, the server learns all the measurements that are shared by at least \(\threshold\) clients (along with any auxiliary data).\footnote{Note that similar approaches were highlighted previously by Bittau et al.~\cite{prochlo}, but various complex issues were left as open problems to solve.}

The aforementioned randomness server runs an oblivious pseudorandom function (OPRF) service that allows clients to receive pseudorandom function evaluations on their measurement and the server OPRF key, without revealing anything about their measurement. The clients use the output as randomness to produce the message that is then sent to the aggregation server. Using this framework allows \tool{} to provide strong privacy guarantees for clients, even if the measurement space has low entropy at the point when the aggregation takes place. The randomness server must be non-colluding with respect to the aggregation server, though these servers never have to communicate directly. Note that OPRF services are already being standardized by the Internet Engineering Task Force (IETF)~\cite{I-D.irtf-cfrg-voprf}, and several open-source implementations already exist.\footnote{\url{https://github.com/cfrg/draft-irtf-cfrg-voprf}}

The full \tool protocol is specified in Section~\ref{sec:scheme}. We also describe an alternative form of \tool{}, ``\toollite{}'', that samples randomness only from the measurement itself. This approach is only suitable for sufficiently random data distributions, but removes the need for a distinct randomness server, further simplifying and reducing the costs of private data collection. See Section~\ref{sec:scheme} for more details, and Section~\ref{sec:discussion} for wider discussion of the security guarantees.

\point{Trust assumptions} While \tool protocol is inherently
multi-server, we note that the collaboration model is categorically
weaker than previous cryptographic approaches such
as~\cite{NSDI:CorBon17,SP:BBCGI21,AsiaCCS:BlaAgu12}, where multiple
servers collaboratively compute the output of the aggregation. In
effect, \tool provides the same trust dynamic as submitting plaintext
measurements to an untrusted server over an anonymizing proxy (which
also provides the randomness server functionality), but with the extra
security guarantee that client measurements are hidden until
\threshold-anonymity is provided, and with very little additional
performance overhead.

\point{Simple cryptography} \tool uses simple, well-established cryptographic
tools, that have been used extensively by non-experts for many years.
Previous proposals either use trusted hardware; non-quantifiable
noise-based approaches; or novel, complex, and poorly understood
cryptographic tools.

\point{Performance} To confirm the practicality of \tool, we present and
report on an open-source Rust implementation.\repourl For
processing server-side aggregation of~\(1,000,000\) client measurements,
\tool requires only~\({20.01}\) seconds of computation and a total of~\({222.21}\)MB of
communication, and client overheads are minimal. Overall, \tool is orders of magnitude cheaper to run than previous systems, see Section~\ref{sec:experimental} for more details.

\point{Standardization} \tool is compatible with the IETF's
proposed framework for devising new privacy-preserving measurement
systems~\cite{ietf-ppm}.

% \point{Limitations of \tool} A limitation of \tool  is that the protocol leaks which clients share measurements, even if the measurements themselves remain hidden. We show that this leakage is similar to previous schemes that attempt to build practical cryptographic mechanisms for threshold aggregation. See Section~\ref{sec:functionality-comparison} for a comparison of the leakage profiles between \tool and previous work.

\subsection{Formal contributions}
We make the following contributions.
\begin{itemize}
    \item The design, systematization, and formalization of the \tool system, and associated privacy goals.
    \item An open-source Rust implementation of \tool, already used in large-scale deployments.
    \item Empirical evaluation of the \tool protocol that showcases performance and simplicity far superior to previous constructions, while ensuring comparable privacy guarantees.
    \item Specific guidance for navigating trade-offs between additional privacy, and simpler deployment scenarios.
\end{itemize}

\section{Overview of Design Goals}
\label{sec:design-goals}
In this section, we clarify the problem statement that we are tackling, along with subsequently a set of design goals and non-goals that we consider.

\subsection{Problem statement}
\label{sec:problem-statement}

\point{Primary goal} We aim to build a system that allows clients to submit measurements as encoded messages to an untrusted aggregation server. This aggregation server should be able to decode and reveal \emph{only} those measurements that are sent by \(\geq \threshold\) clients, where \(\threshold\) is a public parameter chosen by the aggregation server.

\point{Auxiliary data} Clients should be able to send auxiliary data with their measurements, that can differ from client-to-client and is revealed only if the client's measurement satisfies the threshold aggregation policy.

\subsection{Design goals}
\label{sec:motivating-example}
\label{sec:core-constraints}
\label{sec:goals}
We aim to enable privacy-preserving threshold aggregation data collection through a protocol that both i) provides strong privacy guarantees, and ii) is practical for implementation and adoption by a wide range of projects and organizations; everything from small hobbyist projects to Web-scale software. We particularly aim for a solution for projects that are not well served by existing state-of-the-art (which requires non-trivial budgets, difficult-to-achieve trust assumptions and implementation expertise). To assess suitability, the following points and constraints are crucial to bear in mind.

\point{Client privacy} Any protocol should provide formal guarantees of client privacy in a well-understood and coherent security model, with very limited leakage.

\point{Correctness guarantees}
Any solution must provide \emph{correct} aggregation, rather than approximations that rely on receiving very large amounts of client data for providing high utility.

\point{Low financial costs}
Small projects usually run servers in standard cloud-based hardware such as Amazon Web Services (AWS), so financial costs can run up quickly. Thus, we can neither tolerate expensive cryptographic computation nor costly bandwidth consumption. Ideally, we would like aggregation of \(1\) million client measurements to incur a cost of less than \(1\) dollar.

\point{Achievable trust requirements}
Data aggregation procedures that rely on multi-round interactions with a non-colluding partner are expensive to set up, run, and maintain. In other words, there should be a single aggregation server that must not require communication with any non-colluding parties, at least during the aggregation process.

\point{Avoiding trusted hardware}
Running aggregation in trusted hardware platforms, such as secure enclaves (such as Intel SGX) or cloud-based solutions (e.g. Amazon Nitro enclaves\footnote{\url{https://aws.amazon.com/ec2/nitro/}}), are usually prohibitively expensive and potentially vulnerable to attacks~\cite{arxiv:NilBidBro20}. Overall, requiring trusted hardware significantly increases the complexity of any candidate system.

\point{Limiting cryptographic complexity}
Avoiding novel cryptographic procedures, that are typically expensive to run and require significant expertise to implement, allows those with little cryptographic knowledge to implement applications safely. This decreases the risk of disastrous privacy vulnerabilities, and increases auditability of security guarantees.

\subsection{Non-goals}
\label{sec:non-goals}
Furthermore, we make clear that we are not attempting to solve any of the following problems.

\point{Prevention of Sybil attacks} By their very nature, Sybil attacks~\cite{IPTPS:Douceur01}~---~where a malicious aggregation server injects clients into the system that send messages to try and reveal data from honest clients~---~are an unavoidable consequence of building any threshold aggregation system. Therefore, we will not be attempting to provide security for any client measurements that are targeted by such attacks. We will instead provide a security model that restricts the time window in which such attacks can occur (Section~\ref{sec:scheme}). Our solution will also be compatible with any typical higher-layer defenses that are typically used (such as identity-based certification~\cite{IPTPS:Douceur01}).

\point{Leakage-free cryptographic design} All threshold aggregation systems that approach practical performance involve disclosing small amounts of leakage about client measurements that remain hidden. Combined with external public data, this leakage may become more useful in identity-linkage attacks. Rather than preventing leakage entirely, we will instead show that the \tool approach provides a leakage profile that is comparable with recent work in this area (Section~\ref{sec:security-model}).

% \point{Cryptographic novelty} We will not be introducing any novel cryptographic primitives in this work. Our aim is to build a coherent aggregation protocol from as simple primitives as possible. We view this as a strength of our work, as the \tool system design is explicitly designed to be understandable by those with only a minimal background in cryptography. From a security perspective, any primitives that we use have been known to the cryptographic community for many years (and potentially standardized), which means that they convey stronger guarantees against spurious cryptanalysis.

\section{Preliminaries}
\label{sec:background}

\label{sec:prelim}
\label{sec:crypto}
We provide the descriptions of each of the cryptographic primitives that are used for constructing the \tool protocol.

\point{General notation}
We use PPT to describe a probabilistic polynomial time algorithm. We use
\([n]\) to represent the set \(\{1,\ldots,n\}\). We use \(x \| y\) to
denote the concatenation of two binary strings. We write
\(\X \comp \Y\) for (computationally indistinguishable) distributions \(\X\) and \(\Y\) iff the advantage of
distinguishing between \X and \Y for any PPT algorithm is negligible. We write \(\X \stat \Y\) if \X and \Y are (statistically) indistinguishable for any algorithms (even if they run in exponential time). 

\point{Symmetric encryption}
\label{sec:symmenc}
We will assume a symmetric key encryption scheme,
\(\skescheme\), that consists of two algorithms: 
\begin{itemize}
    \item \(c \leftarrow \enc(k,x)\): produces a ciphertext \(c\) as the output of encrypting data \(x\) with key \(k\);
    \item \(x \leftarrow \dec(k,c)\): outputs \(x\) as the decryption of \(c\) under key \(k\).
\end{itemize}
We separately use \(\derivenoinput\) to
denote an algorithm that accepts a seed and a security parameter as
input, and returns a randomly sampled encryption key. We will assume that, for randomly sampled keys, \(\skescheme\) satisfies \(\indcpa\) security.

\point{Secret-sharing}
\label{sec:adss}
We assume the usage of a \(\threshold\)-out-of-\(n\) threshold secret
sharing scheme \(\adssschemeshort{\threshold}{n}\) with information-theoretic security, operating in a finite field \(\FF_p\) for some prime \(p>0\). Such a scheme
consists of two algorithms:
\begin{itemize}
    \item \(s \leftarrow \adssschemeshort{\threshold}{n}.\share{z;r}\): a probabilistic algorithm that
    produces a random \threshold-out-of-\(n\) \emph{share} \(s \in \FF_p\) of \(z\);
    \item \((\bar{z},\bot) \leftarrow \adssschemeshort{\threshold}{n}.\reveal{\{s_i\}_{i \in [\ell]}}\): outputs \(\bar{z}\) when \(\ell \geq \threshold\) and each \(s_i\) is a valid share of \(\bar{z}\), otherwise outputs \(\bot\).
\end{itemize}
For security, we require that any
set of shares smaller than \(\threshold\) is
indistinguishable from a set of random strings.\footnote{As is common
for secret sharing schemes, shared messages must be
sufficiently unpredictable~\cite{PoPETS:BelDaiRog20}.} We call this
property \emph{share privacy}, and is achieved for secret
sharing approaches based on traditional Shamir secret
sharing~\cite{PoPETS:BelDaiRog20}.

\begin{remark}
    Note that the \sharenoinput algorithm remains probabilistic and samples randomness internally when constructing individual shares. In the language of Shamir secret sharing, the explicit randomness input \(r\) is used to derive the coefficients of the \(\kappa-1\) degree polynomial, \(P\). Shares are then derived by sampling a random value \(c \in \FF_p\) and evaluating \(P(c)\).
\end{remark}

\begin{remark}
    We require
    that \(p\) is large enough that randomly sampling values from
    \(\FF_p\) is highly unlikely to lead to collisions.
    Note that the size of \(p\) does not have any bearing on security.
\end{remark}

\point{Oblivious pseudorandom function protocols}
\label{sec:prelim-oprf}
We assume the presence of a verifiable oblivious pseudorandom function (VOPRF) protocol denoted by \(\voprfscheme\).
Oblivious pseudorandom function (OPRF) protocols were first introduced
by Freedman et al.~\cite{TCC:FIPR05}. They enable a client to receive
PRF evaluations from a server, whilst the client input is kept secret,
and nothing is revealed about the server PRF key. Verifiable OPRFs
(VOPRFs) such as that of Jarecki et al.~\cite{AC:JarKiaKra14} provide
clients with the ability to verify (in zero-knowledge) that the server
has evaluated the PRF properly.

Following the description given by Tyagi et al.~\cite{EPRINT:TCRSTW21}, we define a VOPRF, \voprfscheme, to have the following algorithms:

\begin{itemize}
    \item \(\pparams \leftarrow \voprfscheme.\poprfsetup{\secparam}\): a server-side algorithm that produces public parameters \(\pparams\) for the VOPRF;
    \item \((\msk,\mpk) \leftarrow \voprfscheme.\mskgen{\pparams}\): a server-side algorithm that samples a keypair that is compatible with the input parameters \(\pparams\);
    \item \((\rq,\state) \leftarrow \voprfscheme.\req{x}\): a client-side algorithm that produces a request \(\rq\) and some state \(\state\), from some initial input \(x \in \bin^*\);
    \item \(\rp \leftarrow \voprfscheme.\eval{\msk,\rq}\): a server-side algorithm that produces a response \(\rp\) using a secret key \(\msk\), and a client request \(\rq\);
    \item \(y \leftarrow \voprfscheme.\finalize{\mpk,\rp,\state}\): produces the PRF output on \(\msk\) and the input \(x\) encoded in \(\rq\), using the server response \(\rp\), public key \(\mpk\), and client state \(\state\).
\end{itemize}

We assume a VOPRF protocol that follows the standard ideal functionality, as laid out by Albrecht et al.~\cite{PKC:ADDS21}. Such VOPRFs have been shown to exist based on the One-More-Gap-Diffie-Hellman assumption, with security proven in the UC-security model~\cite{AC:JarKiaKra14}.

It should be noted that there are numerous practical use-cases for
(V)OPRF protocols and their
variations~\cite{PoPETS:DGSTV18,C:KLOR20,DITFB,EPRINT:TCRSTW21},
alongside IETF standardization
efforts~\cite{I-D.irtf-cfrg-voprf,I-D.irtf-cfrg-opaque}, and open-source implementations.\footnote{\url{https://github.com/cfrg/draft-irtf-cfrg-voprf}}

\point{Min-entropy} For a distribution \(\D\) over some input space
\(\X\), the min-entropy of \(\D\) is defined as \(\min_{x \in \X}(-\log_2(\prob{X = x}))\).

\subsection{Protocol security model}
\label{sec:malsec}
In Section~\ref{sec:security-model}, we describe an ideal functionality of the threshold aggregation protocol~---~including inputs, outputs, and potential leakage~---~and use it to show that any attack that is possible in the real world protocol is also possible to launch against the ideal world functionality. Intuitively, this proves that the protocol reveals nothing except what is revealed by the function output \emph{plus} a bounded amount of leakage that is output by a specific leakage function.

\point{Protocol security}
The ideal functionality is denoted by
\(\F_{\P}\) for protocol \P. Let \(\inputs_\H\) and \(\inputs_\adv\)
denote the set of inputs chosen by both honest parties and the adversary \(\adv\), respectively. In addition, let \(\real{\P,\adv;\inputs_\adv,\inputs_\H}\) denote the view of the
adversary \(\adv\) in the real protocol, and
\(\ideal{\F_{\P},\sdv,\adv;\inputs_\adv,\inputs_\H}\)
the view of \(\adv\) when simulated by a PPT algorithm \(\sdv\) that interacts with \(\F_{\P}\). We say that \P is secure against \emph{malicious
adversaries} if, for all choices of inputs, the following equation holds:
\begin{multline}
    \real{\P,\adv;\inputs_\adv,\inputs_\H} \\ 
    \comp \ideal{\F_{\P},\sdv,\adv;\inputs_\adv,\inputs_\H},
    \label{eq:malmpcleakage}
\end{multline}
This security model is commonly referred to as proving security in the \emph{real/ideal-world} paradigm.

\point{Leakage}
The leakage function specifies additional information that the adversary may learn during the protocol that is required for completing the simulation. This information mirrors real leakage that occurs during the protocol execution. We denote by \(\leakagenoinput\) the leakage function that takes as input a set of inputs \(\inputs\) (both honest and adversarial), and outputs some leakage \(\leakage{\inputs}\).

\section{\texorpdfstring{The \tool Protocol Framework}{The STAR Protocol Framework}}
\label{sec:scheme}
\subsection{Notation} 
\label{sec:notation}
\noindent \underline{Participants and protocol parameters:}
\begin{itemize}
    \itemsep0em
    \parskip0em
    \parsep0em
    \item We use \(\P\) to refer to the \tool protocol, and \(\widetilde{\P}\) to refer to \toollite (see Section~\ref{sec:protocol-desc}).
    \item \(\threshold\) is the threshold used for performing aggregation.
    \item \(\cdv\) is the set of all clients \(\{\CC_i\}_{i \in [n]}\).
    \item \(\SS\) is the aggregation server.
    \item \(\OO\) is the randomness server in \(\P\).
\end{itemize}

\noindent \underline{General notation:}
\begin{itemize}
    \itemsep0em
    \parskip0em
    \parsep0em
    \item \(\D\) be the distribution over universe \(\udv\), that clients sample their measurements from.
    \item \((c_i,s_i,t_i)\) is the message sent by \(\CC_i\), as defined in Figure~\ref{fig:protocol}.
    \item \(\X\) is the set of all measurements received by \(\SS\), and let \(\X_\H\) (\(\X_\adv\)) denote the subsets of measurements received from honest (adversarial) clients.
    \item Let \(\E_1 = (x_1,\aux_{x_1},\threshold_1),\ldots,\E_\ell = (x_\ell,\aux_{x_\ell},\threshold_\ell)\) correspond to each of the \(\ell\) unique measurements \(x_i\) sent by clients to \(\SS\). Therefore, \(\aux_{x_i}\) is the collection of all auxiliary data associated with each message containing the measurement \(x_i\), and \(\threshold_i\) is the number of such messages that are received by \(\SS\).
    \item \(\Y\) is the set containing each \(\E_\iota\) where \(\threshold_\iota \geq \threshold\) that is output to \SS.
\end{itemize}

\noindent \underline{Cryptographic tools:}
\begin{itemize}
    \itemsep0em
    \parskip0em
    \parsep0em
    \item \(\voprfschemeshort\) is a \voprf{} (Section~\ref{sec:crypto}).
    \item \((\msk,\mpk)\) is the keypair of \OO for \voprfschemeshort.
    \item \(\skeschemeshort\) is a symmetric encryption scheme
    satisfying \(\indcpa\) security (Section~\ref{sec:crypto}).
    \item \(\adssschemeshort{\threshold}{n}\) is a \adss{\threshold}{n}
    scheme, and let \(\FF_p\) be the associated finite field with order \(p \in \ZZ\) (Section~\ref{sec:crypto}).
    \item \(\adv\) is a malicious PPT adversary.
    \item \(\sdv\) is a PPT simulator.
    \item \(\F_{\P}\) is the ideal functionality corresponding to protocol
    \(\P\), and \(\leakagenoinput\) is the leakage function for \P (Section~\ref{sec:security-model}).
    \item \(\F_{\voprfschemeshort}\) is the ideal functionality corresponding to \voprfschemeshort.
\end{itemize}

\subsection{Design space} We assume a large universe of elements \(\M\)
(e.g., bitstrings of \(\geq 64\) bits) representing potential
\emph{measurements} that clients send to a single, untrusted aggregation
server. For example, such measurements may include profile information
about a user (e.g. browser user-agent), or the set of applications
installed on a device. Clients may \emph{optionally} send arbitrary
additional data with their measurement.

A single encoded measurement is sent during an \emph{epoch} by each
available client. The aggregation server should be able to reveal all
those encoded measurements (and any associated data) that are received
at least \(\threshold\) times. The threshold \(\threshold \geq 1\) is agreed publicly from the outset.

\begin{figure}[!t]
    \raggedright
    \fbox{%
        \scalebox{0.67}{%
            \procedure[codesize=\large]{\tool \textsf{Randomness phase}}{\\[-2ex]
                \boxed{\CC_i(\pparams,x_i,\mpk)}
                \<\<
                \boxed{\OO(\pparams,\msk,\mpk)}\\
                (\rq_i,\state_i) \leftarrow
                \voprfschemeshort.\req{x_i} \< \sendmessageright{length=2cm,top={$\rq_i$}} \< \\
                \< \sendmessageleft{length=2cm,top={$\rp_i$}} \< \rp_i \leftarrow \voprfschemeshort.\eval{\msk,\rq_i} \\
                r_i \leftarrow
                \voprfschemeshort.\finalize{\mpk,\rp_i,\state}
                \<\< \\
                r_{i,1} \| r_{i,2} \| r_{i,3} \leftarrow r_i
                \<\<
                \\[-2ex]
                \\[][\hline]
                \\[-2ex]
                \underline{\text{Key rotation}}\<\< \\
                \<\< \pparams \leftarrow \voprfschemeshort.\poprfsetup{\secparam}\\
                \< \sendmessageleft{length=2cm,top={$(\pparams,\mpk)$}} \<
                (\msk,\mpk) \leftarrow \voprfschemeshort.\mskgen{\pparams}
            }
        }
    }\vspace{0.05cm}
    \begin{pcvstack}[space=0.25cm]
        \procedure[bodylinesep=1mm]{\textsf{\tool Message phase}}{
            \textbf{Inputs}: x_i,\aux_i,(r_{i,j})_{j
            \in [3]},\threshold\\
            \textbf{Outputs}: \text{Client \tool message}\\
            \pcln K_i \leftarrow \derive{r_{i,1},\secparam} \\
            \pcln s_i \leftarrow
            \adssschemeshort{\threshold}{n}.\share{r_{i,1};r_{i,2}} \\
            \pcln c_i \leftarrow \skeschemeshort.\enc(K_i,x_i \|
            \aux_i) \\
            \pcln t_i \leftarrow r_{i,3}\\
            \pcln \pcreturn (c_i,s_i,t_i)
        }
        \procedure[bodylinesep=1mm]{\textsf{\tool Aggregation phase}}{
            \textbf{Inputs}: n\ \text{Client \tool messages}, \threshold\\
            \textbf{Outputs}: \text{List of measurements that were sent} \geq \threshold  \text{ times}\\
            \pcln \Y = [];\\
            \pcln \pcforeach \E_{\iota} = \{(c_j,s_j,t_j) \mid (t_j =
            t_\iota\,\forall\,j)\}: \\
            \pcln \pcind \pcif |\E_\iota| < \threshold: \pcreturn
            \perp \\
            \pcln \pcind (\{c_\iota\},\{s_\iota\}) \leftarrow
            \{\{c_\iota\}.\push{c_j},\,\{s_\iota\}.\push{s_j} \mid (c_j,s_j) \in
            \E_\iota\} \\
            \pcln \pcind r_{\iota,1} \leftarrow
            \adssschemeshort{\threshold}{n}.\reveal{\{s_\iota\}} \\
            \pcln \pcind K_\iota \leftarrow \derive{r_{\iota,1},\secparam} \\
            \pcln \pcind \pcforeach c_j \in \{c_\iota\}: \\
            \pcln \pcind \pcind x_\iota \| \aux_j
            \leftarrow \skeschemeshort.\dec(K_\iota,c_j) \\
            \pcln \pcind \pcind \pcif (\iota \neq 0)
            \wedgespace (x_\iota \neq x_{\iota-1}): \\
            \pcln \pcind \pcind \pcind \pcreturn \perp \\
            \pcln \pcind \pcind \Y[x_\iota].\push{\aux_{j}} \\
            \pcln \pcreturn \Y
        }
    \end{pcvstack}
    \caption{ The \tool protocol for performing threshold aggregation of
        measurements. In the \textsf{Randomness
        phase}, clients sample VOPRF randomness from \OO, and \OO rotates their VOPRF keypair at regular intervals. The
        \textsf{Message phase} sees clients construct an encoded
        message corresponding to their
        measurement. In the \textsf{Aggregation phase}, \SS
        receives encoded messages from clients and learns those measurements (and associated data) that are sent by
        \(\geq \threshold\) clients. }
    \label{fig:protocol}
\end{figure}

\subsection{\texorpdfstring{\tool}{STAR} protocol}
\label{sec:protocol-desc}
The \tool protocol is based upon the principle that clients sharing a
measurement can devise compatible secret shares for a
\(\adss{\threshold}{n}\) scheme. Such shares could then be combined to
reveal the measurement itself (and optionally any additional data that
they send) by an untrusted aggregation server. Once
\(\threshold\) clients send a share of the same value, the server
will be able to recover the hidden value (and any additional data
that is sent).

The algorithmic description of the \tool protocol is given in Figure~\ref{fig:protocol}. We provide a description below as an overview of the entire exchange.

\point{Randomness phase: \tool}
Firstly, each client interacts with the randomness server, \OO, to learn correlated randomness for their measurement \(x_i\). Essentially, the client operates as the client in the \voprf protocol with input \(x_i\), and the randomness server answers the query and returns the result to the client. Note that the client must also possess the public parameters, \pparams, and the public key, \mpk, that \OO produces. The client, after processing the \voprf output to receive \(r_i \in \bin^{3\omega}\) for some \(\omega > 0 \in \ZZ\), now has the result \((x_i, r_i)\). Note that any client that shares the measurement \(x_i\) will also receive the same output \(r_i\). See Section~\ref{sec:prelim-oprf} for a description of the VOPRF exchange. Note that the randomness server should periodically rotate their VOPRF keypair to improve client privacy guarantees and clients should be able to download the new public key data accordingly, see Section~\ref{sec:security-considerations} for more details.

\point{Randomness phase: \toollite}
It is possible to construct a version of \tool, known as \toollite, that provides weaker security guarantees, in favor of dropping the requirement for the randomness server (which can leads to a much simpler practical deployment). The client in the \toollite protocol simply samples \(r_i\) directly from their measurement (for example \(r_i \leftarrow H(x_i)\), where \(H\) is a random-oracle model hash function) before proceeding directly to the message phase. The \toollite protocol only retains security when client measurements are sampled from a suitable high-entropy distribution, see Section~\ref{sec:discussion} for more discussion.

\point{Message phase} The message construction phase consists of the following steps.
\begin{enumerate}
    \item The client with \((x_i,r_i)\) parses \(r_i\) into three parts \((r_{i,1},r_{i,2},r_{i,3}) \in (\bin^\omega \times \bin^\omega \times \bin^\omega)\).\footnote{This can be done for example by running \(r_{i,j} = H(r_i \| j)\), for a random-oracle model hash function \(H\).}
    \item They derive a symmetric key \(K_i\) using a pseudorandom generator where \(r_{i,1}\) is used as the seed.
    \item They construct a random share \(s_i\) of \(r_{i,1}\) using a \(\threshold\)-out-of-\(n\) secret-sharing scheme, using \(r_{i,2}\) as explicit randomness that is used in the share generation process.
    \item They construct the ciphertext \(c_i\) as the encryption of their measurement \(x_i\), and any auxiliary data that they would like to attach, using a symmetric encryption scheme with the previously-derived key \(K_i\).
    \item Finally, they construct their message as \((c_i,s_i,t_i)\), where \(t_i = r_{i,3}\).
\end{enumerate}

As mentioned previously in Section~\ref{sec:adss}, the construction of secret shares in step 3 is a probabilistic algorithm. For Shamir's secret-sharing, the randomness \(r_{i,2}\) is only used to agree on a set of polynomial coefficients, and each client individually samples a random polynomial evaluation to create their share.

\point{Aggregation phase} In the final aggregation phase the aggregation server receives a message from each of \(n\) clients, and learns which of the encoded measurements are shared by at least \threshold clients. The steps are as follows.
\begin{enumerate}
    \item The aggregation server groups together messages that share the same \(t_\iota\) value into subsets \(\E_\iota\).
    \item They discard and subsets with fewer then \(\threshold\) messages, and then do the following for each remaining subset: 
    \begin{enumerate}
        \item runs the share recovery algorithm on the collection of share values \(\{s_\iota\} \in \E_\iota\) to output \(r_{\iota,1}\);
        \item derives the encryption key \(K_\iota\) from \(r_{\iota,1}\);
        \item decrypts each client ciphertext \(c_j\) using \(K_\iota\), and groups together the measurement \(x_\iota\) with the list of the auxiliary data objects, \(\aux_j\), sent by each client.
    \end{enumerate}
    \item Finally, the aggregation server creates a list \(\Y\) of all measurements \(x_\iota\) (along with the attached auxiliary data) that satisfy the threshold \threshold, and outputs \(\Y\).
\end{enumerate}

\subsection{Security Considerations}
\label{sec:security-considerations}
We detail a series of considerations related to the security of the \tool protocol design. The formal security model that we will use for proving security is given in Section~\ref{sec:security-model}, and the proofs are given in Appendix~\ref{app:cryptographic-guarantees}.

\point{Communication between servers} Note that the randomness and aggregation servers only communicate with the clients in the
system, and only one performs the eventual aggregation. This is a
significant improvement on existing multi-server solutions for threshold
aggregation, where the servers are required to communicate with each
other for processing the results of aggregation. Requiring communication between servers quickly drives up costs for both server operators, and tangibly weakens the extent to which both servers are non-colluding. This is because the server operators will have to work together to ensure that their servers can cooperate.

\point{Leakage}
The leakage in the \tool protocol amounts to the aggregation server learning which clients share the same measurement~---~regardless of whether the measurement is kept hidden or not. Similarly, the adversary could launch a ``Sybil'' attack by establishing/corrupting clients with specifically-chosen measurements. As mentioned in Section~\ref{sec:non-goals}, we consider prevention of Sybil attacks a non-goal, since all such threshold aggregation protocols are vulnerable to such attacks. However, we encode the possibility for an adversarial \SS to make use of this leakage into the formal leakage function that is defined as part of our security model in Section~\ref{sec:security-model}.

\point{Randomness server key rotations}
The usage of the randomness server in \tool ensures that an adversarial
aggregation server must communicate with the randomness server to launch
attacks on client inputs, but it does not immediately provide security
to low-entropy inputs. Therefore, we consider a security model where
clients sample randomness in epoch \(\mtag\), and send their encoded
measurement in epoch \(\mtag+1\), after the randomness server has
performed a key rotation. This limits the aggregation server to only
launch online attacks on client inputs before epoch \(\mtag+1\), having
not yet seen any client messages, or observed any leakage. Once this key
is deleted, it is not possible to launch queries that attempt to
identify hidden client values. Moreover, by rotating this key before the
aggregation phase takes place, this ensures that the \SS is only able to
make use of any leakage that may occur \emph{before} they witness any
client measurements. 

In the formal security model defined in Section~\ref{sec:security-model}, we encode this by forcing the adversary \SS to specify up front which values they would like to leak. Importantly, this disables the potential for an adversary to launch a targeted attack based on client identity, or any observed leakage.

\point{Predictable input distributions}
Practical use-cases of \tool require that client messages remain somewhat unpredictable during the randomness phase of the protocol. If measurements are predictable, then the aggregation server may launch queries against the randomness server for all such values during this phase, and then use the leakage to learn which clients are sending predictable values, even if fewer than \(\threshold\) clients send them. The advantage of \tool (as opposed to \toollite) is that this attack can only be carried out during the time-limited randomness phase (before the key rotation occurs), and that the attack must be carried out online. This facilitates usage of extra external protection measures at the randomness server, such as identity-based rate-limiting and verification, to make such attacks even more expensive.

\point{Additional data} Before the protocol
begins, \(\SS\) should inform clients of the maximum length of
the additional data that should be sent. If \(\aux_i\) is not equal to
that length, then it must be truncated or padded depending on whether it
is too long or short, respectively. We make no
guarantees on the shape of auxiliary data for client measurements.

\point{Hardening against local attacks in \tool} All hash function
invocations in \tool can be replaced with functions that are
deliberately slower primitives, such
as PBKDF2~\cite{RFC2898} and scrypt~\cite{RFC7914}. Such functions are
used in applications handling passwords that hope to provide additional
security against password-cracking adversaries. This change only impacts
client computation in a small way, and would increase the difficulty for
any adversarial aggregation server trying to reverse client encoded
measurements. Moreover, such changes similarly increase the difficulty
of attacks in case of a breakdown in the trust model used in \toolfull,
or if using \toollite.

\subsection{Reducing Leakage Via Oblivious Proxies}
\label{sec:high-level-privacy}

\point{Identity leakage}
As with many previous designs of threshold aggregation protocols, \tool
produces a quantifiable amount of leakage. However, the link between
client identity and their input is unbroken.

In some applications maintaining this link is useful. Consider an
aggregation server that is attempting to learn which clients may be part
of a fraudulent botnet of a threshold size, by having clients submit
information about their browser profile. In such cases, it is essential
to link client identity to their sent messages, so that the aggregation
server can subsequently disqualify malicious clients.

However, if an aggregation server is merely trying to learn client
diagnostic information, it is unlikely that maintaining this link is
useful or necessary.

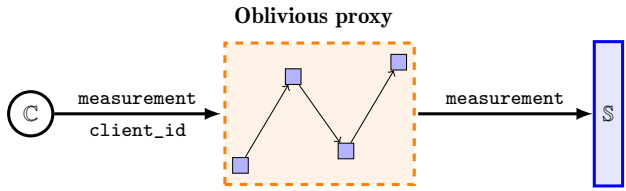
\begin{figure}[t]
    \centering
    \scalebox{0.96}{%
        \begin{tikzpicture}
            \node[fill=blue!30,draw=black] (p1) {};
            \node[fill=blue!30,draw=black,above right=of p1,xshift=-0.5cm] (p2) {};
            \node[fill=blue!30,draw=black,below right=of p2,xshift=-0.5cm,yshift=0.2cm] (p3) {};
            \node[fill=blue!30,draw=black,above right=of p3,xshift=-0.5cm] (p4) {};
            \draw[black,->] (p1) -- (p2);
            \draw[black,->] (p2) -- (p3);
            \draw[black,->] (p3) -- (p4);
            \begin{scope}[on background layer]
                \draw[orange,fill=orange!10,very thick,dashed] ($(p1.south west)+(-0.1,-0.15)$)  rectangle ($(p4.north east)+(0.1,0.15)$);
            \end{scope}
            \node at ($(p2.north west)+(0.4,0.7)$) (title) {\textbf{Oblivious proxy}};
            \node at ($(p1)!0.5!(p4)$) (N) {};
            \node[circle,draw=black,very thick,xshift=-4cm] at (N) (C) {\(\CC\)};
            \node[rectangle,draw=blue,very thick,fill=blue!10,minimum height=2cm,xshift=4cm] at (N) (S) {\(\SS\)};
            \draw[black,->,very thick,>=latex] (C) -- node[above] {\small \texttt{measurement}} node[below] {\small \texttt{client\_id}} ($(N.west)+(-1.25cm,0)$);
            \draw[black,->,very thick,>=latex] ($(N.east)+(1.25cm,0)$) -- node[above] {\small \texttt{measurement}} (S);
        \end{tikzpicture}
    }
    \caption{Oblivious proxy for submitting client (\(\CC\)) measurements to the aggregation server (\(\SS\)).}
    \label{fig:oblivious-proxy}
\end{figure}

\point{Oblivious proxies}
One method for eliminating such leakage in \tool is by using tools for
performing anonymous value submission at the
application-layer~---~destroying the link between client identity and
their messages. For example, by submitting measurements via an
oblivious/anonymizing proxy that strips client identifying information (such as IP addresses)
from HTTP requests containing client measurements, the aggregation
server learns nothing about the client identity (Figure~\ref{fig:oblivious-proxy}). Well-known tools exist
for this purpose such as Tor\footnote{https://www.torproject.org/} (or
certain VPNs) can be used. However, using Tor
comes with well-known performance overheads that would slow down client
requests in \tool considerably~\cite{EPRINT:SerHogDev21}.

\begin{figure}[t]
    \centering
    \scalebox{0.93}{%
        \begin{tikzpicture}
            \node[circle,draw=black,very thick] (C) {\(\CC\)};
            \node[circle,draw=black,very thick,right=of C,xshift=2.5cm,fill=orange!10] (OP) {\(\OO\)};
            \node[rectangle,draw=blue,very thick,fill=blue!10,minimum height=2cm,right=of OP,xshift=2.5cm] (S) {\(\SS\)};
            \draw[black,->,very thick,>=latex] (C) -- node[above,yshift=0.2cm] (A) {\small \texttt{measurement}} node[below,yshift=-0.2cm] (B) {\small \texttt{client\_id}} (OP);
            \draw[black,->,very thick,>=latex] (OP) -- node[above,yshift=0.2cm] (C) {\small \texttt{measurement}} (S);
            \begin{scope}[on background layer]
                \draw[orange,fill=orange!10,very thick,dashed] ($(A.north west)+(-0.2,1)$)  rectangle ($(B.south east)+(0.365,0)$);
                \draw[black,fill=green!10,very thick,dashed] ($(A.north west)+(-0.1,0.4)$)  rectangle ($(A.south east)+(0.1,0)$);
                \node at ($(A.north west)+(0.6,0.15)$) (tls) {\textbf{\underline{HTTPS}}};
                \node at ($(A.north west)+(0.4,0.7)$) (hpke) {\textbf{\underline{HPKE}}};
                \draw[black,fill=green!10,very thick,dashed] ($(C.north west)+(-0.1,0.4)$)  rectangle ($(C.south east)+(0.1,0)$);
                \node at ($(C.north west)+(0.6,0.15)$) (tls) {\textbf{\underline{HTTPS}}};
            \end{scope}
        \end{tikzpicture}
    } \caption{Oblivious HTTP flow including usage of hybrid public key
    encryption (HPKE) for message encapsulation. Here, \(\OO\) is the
    \emph{proxy resource}~\cite{I-D.thomson-http-oblivious}. This entity
    can be implemented in \toolfull using the randomness server \OO,
    since the client messages are protected with TLS.}
    \label{fig:oblivious-http}
\end{figure}
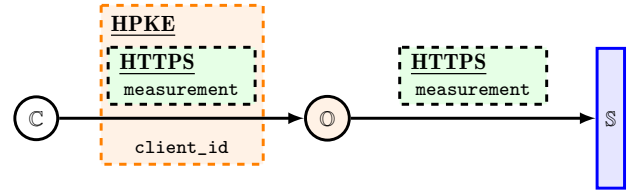
\point{Oblivious HTTP} An alternative mechanism known as Oblivious
HTTP (OHTTP) that has been proposed as a draft standard to the
IETF~\cite{I-D.thomson-http-oblivious} performs similar anonymization of
HTTP requests as Tor, but with fewer intermediate hops~---~promising a smaller
performance overhead. The oblivious proxy is a single party known as the
\emph{proxy resource}, and the aggregation server plays the part of a
\emph{target resource} that receives the client message~\cite{I-D.thomson-http-oblivious}.

Figure~\ref{fig:oblivious-http} provides a diagrammatic representation
of the OHTTP flow in the context of \tool. In essence, the client
encapsulates a HTTP(S) request containing their message to the aggregation
server using \emph{hybrid public key encryption}
(HPKE)~\cite{I-D.irtf-cfrg-hpke}, where encapsulation is performed under
the public key of the oblivious proxy. The client sends this
encapsulated message as the body of a separate HTTP(S) request to the
oblivious proxy. The proxy decapsulates the message and forwards it on
to the aggregation server, without including any client identifying
information.  

Since client messages to the aggregation server are protected by TLS, the
oblivious proxy has no way of reading the client messages. As a result,
this oblivious proxy can be instantiated using the existing randomness
server \OO in \toolfull without compromising any of the security goals,
and without requiring any additional non-colluding parties. Note that
this means that the randomness server must explicitly send a message to
the aggregation server, whereas the original \tool protocol requires no
communication between these two entities. This communication is minimal
and not related at all to the cryptographic logic that is run in the
aggregation server. Even so, operators that prefer to avoid any
communication taking place between these servers can simply submit data
over existing anonymizing proxies like Tor, or can require the
oblivious proxy to be run by a different party.

The Oblivious HTTP Internet standards draft defines specific guarantees
that must be upheld by the anonymizing proxy, as well as request
formats~\cite[Appendix A]{I-D.thomson-http-oblivious}. Such proxies are
already intended to be standardized by the IETF, and to be run by independent
entities\footnote{IETF OHAI:
\url{https://datatracker.ietf.org/group/ohai/about/}} for
privacy-preserving measurement aggregation systems~\cite{ietf-ppm}.

\subsection{Formal Security Model}
\label{sec:security-model}
We now provide the security model for establishing the security of \toolfull. See Appendix~\ref{app:cryptographic-guarantees} for a sequence of theorems and proofs that guarantee the security and correctness of the \tool protocol, with respect to the following model.

\point{Ideal functionality}
The ideal functionality below represents the inputs, outputs, and internal steps of the threshold aggregation functionality. We will write \(\F_{\P}\) to denote this functionality, where \P is the \tool protocol.
\begin{itemize}
    \item Participants: aggregation server \SS, randomness server \OO, clients \(\{\CC_i\}_{i \in [n]}\).
    \item Public parameters: upper bound on \(n\).
    \item Functionality:
    \begin{itemize}
        \item \OO inputs the VOPRF parameters \pparams, and keypair \((\msk,\mpk)\).
        \item Each client \(\CC_i\) (\(i \in [n]\)) provides their input \((x_i,\aux_i)\).
        \item Let \(\E_\iota = \left\{(x_\iota,\{\aux_j\}_{j \in J},\threshold_\iota)\ :\ (J \subseteq [n])\,\wedge\,(x_j = x_\iota) \right\}\) for each unique \(x_\iota\) received, where \(\threshold_\iota = |\{\aux_j\}|\) is the number of client measurements collected in \(\E_\iota\).
        \item Let \(\Y\) be an empty map.
        \item For each \(\E_\iota\) where \(\threshold_\iota \geq \threshold\), set \(\Y[x_\iota] = \{\aux_j\}_{j \in J}\).
        \item Output \(\Y\) to \(\SS\), output \(\{\F_{\voprfschemeshort}(\pparams,\mpk,\msk,x_i)\}_{i \in [n]}\) to each \(\CC_i\) (receiving the client output) and \OO (receiving the server output).
    \end{itemize}
\end{itemize}
Overall, this ideal functionality captures the fact that the aggregation server learns all client measurements that are sent by at least \(\threshold\) clients. The randomness server learns what it would normally learn during the VOPRF exchange, and each client learns nothing.\footnote{For \toollite the ideal functionality does not include inputs or outputs for the randomness server.}

\point{Leakage function}
We use the leakage function (\leakagenoinput) defined below to account for additional protocol leakage that occurs while running \tool.
Assume that the aggregation server \SS, and some subset \(\T \subset \cdv\) of all clients is controlled by an adversary \(\adv\). The view of \(\adv\) can be simulated using the following leakage function.
\begin{itemize}
    \item Receive \(\W \leftarrow \adv\), a set of disqualified clients specified by \(\adv\).
    \item Receive \(\X_\adv \leftarrow \adv\), a set of input measurements specified by \(\adv\)
    \item Let \(\Q = \cdv \setminus \W\) be the set of remaining honest clients.
    \item Receive \((x_i,\aux_i)\) from each \(\CC_i \in \cdv\).
    \item Partition the set \(\{(x_i,\aux_i)\}_{i \in [|\Q|]} \cup \X_\adv\) into \(\N_1,\ldots,\N_\ell\), where \(\N_\iota\) is the set of all pairs that share the same measurement \(x_\iota\) (for \(\ell\) unique measurements).
    \item Leak \(|\N_\iota|\) to \adv, for each \(\iota \in [\ell]\).
\end{itemize}
We write \(\leakage{\X_\H}\), where \(\X_\H\) is the set of all measurements received from honest clients, to denote the output of \(\leakagenoinput\) on \(\X_\H\). Overall, this leakage function captures the fact that an adversary that controls \SS learns the cardinality of clients that share each unique measurement that is received. 

Note that the leakage function explicitly does not capture the notion of client identity, since we assume that client measurements are submitted anonymously. This can be achieved using various practical solutions (Section~\ref{sec:high-level-privacy}).\footnote{The leakage function could also be trivially updated to capture this additional leakage, if anonymous submission is not possible.}

\section{Functionality and Leakage Comparison}
\label{sec:functionality-comparison}
We compare \tool with prior constructions of private threshold aggregation schemes, specifically with respect to functionality and leakage profiles. See Section~\ref{sec:related-work} for a complete discussion on previous work related to this topic.

\subsection{Ideal functionality}
\label{sec:ideal-comparison}
\begin{figure*}[!t]
  \centering
  \scalebox{0.75}{%
      \begin{tabular}{
        l%
        >{\raggedleft\arraybackslash}p{18mm}%
        >{\raggedleft\arraybackslash}p{18mm}%
        >{\raggedleft\arraybackslash}p{18mm}%
        >{\raggedleft\arraybackslash}p{18mm}%
        >{\raggedleft\arraybackslash}p{18mm}%
        >{\raggedleft\arraybackslash}p{18mm}%
        >{\raggedleft\arraybackslash}p{18mm}%
        >{\raggedleft\arraybackslash}p{18mm}%
      }
          \toprule
          Protocol & Single-round\newline interaction\newline with clients & Bandwidth & Client\newline computation & Aggregation\newline computation & Single-server aggregation & Associated\newline data & Negligible correctness errors & Fail-safety \\
          \midrule
          Proxy-based shuffling~\cite{prochlo,Chaum81,CCS:Neff01} & \goodcell{\cmark} & \goodcell{\(O(n)\)} & \goodcell{\(O(1)\)} & \goodcell{\(O(n)\)}  & \badcell{\xmark} & \goodcell{\cmark} & \goodcell{\cmark} & \badcell{\xmark} \\
          Kissner et al.~\cite{C:KisSon05} & \badcell{\xmark} & \badcell{\(O(mn\secpar)\)} & \badcell{\(O(n^2)\)} & \badcell{\(O(mn\secpar)\)} & \goodcell{\cmark} & \badcell{\xmark} & \goodcell{\cmark} & \badcell{\xmark} \\
          Blanton et al.~\cite{AsiaCCS:BlaAgu12} & \goodcell{\cmark} & \badcell{\(O(mn\secpar)\)} & \badcell{\(O(n^2)\)} & \badcell{\(O(mn^2\secpar)\)} & \badcell{\xmark} & \badcell{\xmark} & \goodcell{\cmark} & \badcell{\xmark} \\
          Randomized response~\cite{JMLR:BNST20,STOC:BasSmi15,ACMTA:BNS19,CCS:QYYKXR16,PMLR:ZKMSL20} & \badcell{\xmark} & \goodcell{\(O(n\secpar)\)} & \goodcell{\(O(1)\)} & \goodcell{\(O(n)\)} & \goodcell{\cmark} & \badcell{\xmark} & \badcell{\xmark} & \badcell{\cmark} \\
          Boneh et al.~\cite{SP:BBCGI21} & \goodcell{\cmark} & \badcell{\(O(mn\secpar)\)} & \goodcell{\(O(\secpar)\)} & \avgcell{\(O(mn\secpar\kappa)\)} & \badcell{\xmark} & \badcell{\xmark} & \goodcell{\cmark} & \badcell{\xmark} \\
          \tool (Section~\ref{sec:scheme}) & \goodcell{\cmark} & \goodcell{\(O(n\secpar)\)} & \goodcell{\(O(\secpar)\)} & \goodcell{\(O(n\secpar\threshold^2)\)} & \goodcell{\cmark} & \goodcell{\cmark} & \goodcell{\cmark} & \goodcell{\cmark} \\\bottomrule
      \end{tabular}
      } \caption{Coarse-grained comparison of \tool against previous
  work. We use \(\secpar\) to denote the security parameter, \(n =
  |\cdv|\) to denote the number of clients, and \(m\) to denote the
  number of servers that are used in multi-server settings. Note that we ignore generic MPC techniques for computing
  threshold aggregation due to well-established performance
  limitations~\cite{CCS:DoeShe17}. We also do not include Prio-like protocols~\cite{NSDI:CorBon17,EPRINT:AGJO21} as they are not compatible with string-based data.}
  \label{table:comparison}
\end{figure*}

A coarse-grained comparison of the functionality provided in \tool with previous approaches is given in Figure~\ref{table:comparison}. All performance costs are
asymptotic, see Section~\ref{sec:experimental} for the concrete costs of
running \tool{}. Overall, the solutions that offer the closest functionality, while still retaining close to practical performance, are the private heavy hitters protocols of~\cite{SP:BBCGI21,ACMTA:BNS19,PMLR:ZKMSL20,JMLR:BNST20,STOC:BasSmi15,CCS:QYYKXR16}. Protocols based on MPC involve very complex cryptographic implementations and expensive overheads~\cite{CCS:DoeShe17}. Protocols that utilize trusted proxies and hardware require clients to place trust in computing platforms and entities that are not immune to security failures~\cite{arxiv:NilBidBro20}.

\subsection{Leakage}
\label{sec:leakage-comparison}
An optimal solution to the threshold aggregation problem would provide information that can be derived from the output of the ideal functionality alone. In other words, only those measurements that are received from \(\threshold\) clients. While some schemes are able to achieve this notion~\cite{prochlo,NSDI:CorBon17,C:KisSon05,AsiaCCS:BlaAgu12}, they typically fall short of providing practical solutions. 

Recent approaches for efficiently learning \(\threshold\)-heavy-hitters~\cite{SP:BBCGI21,ACMTA:BNS19,PMLR:ZKMSL20,JMLR:BNST20,STOC:BasSmi15,CCS:QYYKXR16} incorporate some amount of leakage, that provides additional information to the adversary. Specifically, each scheme leaks all the \(\threshold\)-heavy-hitting prefixes of the eventual \(\threshold\)-heavy-hitter measurements. As an example, consider clients that sent a measurement corresponding to their birth country. Assume that \(\threshold=4\), and that five clients send \texttt{"United States of America"}, four send \texttt{"United Kingdom"}, and three send \texttt{"United Arab Emirates"}. Then the ideal functionality suggests that the aggregation server should only learn that five clients sent \texttt{"United States of America"}, and four sent \texttt{"United Kingdom"}. However, additional leakage informs the server that twelve clients sent the prefix \texttt{"United"}. While such leakage may not always be useful, in this example this effectively leaks how many clients also sent the answer \texttt{"United Arab Emirates"} (since no other country begins with the \texttt{"United"} prefix).

While \tool avoids prefix-based leakage, it leaks the subsets of clients that share equivalent measurements. In other words, the server can separate client messages into groups that all share the same measurement. This can be especially damaging in situations where the adversary launches a ``Sybil'' attack and injects their own measurements to try and learn how many times the same measurement is submitted. As mentioned previously, ``Sybil'' attacks are ultimately possible against any threshold aggregation scheme (even those that do not permit any leakage), and so this is not unique to \tool. Separately, such leakage could allow for measurement inference-based attacks that utilize the counts of each received message to attempt to infer encoded measurements.

Finally, it should be noted that the single-server aggregation mechanisms of \tool and those based on randomized response~\cite{JMLR:BNST20,STOC:BasSmi15,ACMTA:BNS19,CCS:QYYKXR16,PMLR:ZKMSL20} naturally allow linking client messages to revealed measurements. Such leakage can be eliminated using anonymizing proxies for submitting client messages (Section~\ref{sec:high-level-privacy}). This approach has already been recommended for submitting measurements as part of ongoing standardization work in this area~\cite{ietf-ppm}.

\section{Performance Evaluation}
\label{sec:experimental}
We provide an open-source Rust implementation of all the necessary
components for establishing the performance of \tool.\repourl We benchmark the runtimes for both
constructing client messages, and running the server aggregation
process. We estimate the overall bandwidth costs as a result of client's
interacting with both the aggregation and randomness servers. Finally,
we provide runtimes and communication costs for performing anonymization
of \tool{} messages via the Oblivious HTTP
framework~\cite{I-D.thomson-http-oblivious}. Overall, \tool is exceptionally efficient, even when processing 1 million measurements, and orders of magnitude cheaper than competing approaches.

\subsection{Implementation Details}
\point{Secret-sharing implementation}
Our secret sharing implementation is based on the Adept Secret Sharing (ADSS)
framework developed by Bellare et al.~\cite{PoPETS:BelDaiRog20} for
achieving stronger guarantees on privacy and authenticity of shares.

As noted previously, we require implementation of a prime-order finite field
for secret sharing that is large enough to make the occurrence of
collisions a low probability event to ensure correctness. We choose two
prime-order fields~---~one with a modulus of~255 bits in length
(\(\FF_{255}\)), and one that is~129-bits (\(\FF_{129}\))~---~and
provide performance for both. In a practical sense, we consider the
change of collisions in either field to be negligible. We assume that
all inputs that are shared are \(16\) bytes in length (randomness for
deriving symmetric encryption keys), so that they can be stored in a
single share polynomial for either choice of finite field.

Finally, we note that secret share recovery uses only a subset of of
 \(\threshold\) shares. For example, if we receive~\(200\) shares for a
 given measurement, with \(\threshold = 100\), we will only perform
 recovery using a subset of \(100\) shares. This means that we do not
 check whether all client shares are well-formed, but we do perform
 checks on the decrypted result for \emph{all} of them.

\point{Oblivious HTTP proxy} We use an open-source Rust implementation
for constructing an Oblivious HTTP
proxy\footnote{\url{https://github.com/martinthomson/ohttp}} that is
compliant with the most recent IETF standards
draft~\cite{I-D.thomson-http-oblivious}, as described in
Section~\ref{sec:high-level-privacy}. Our setup assumes that client
messages are sent via a \emph{proxy resource}, run by \(\OO\), to a
\emph{target resource}, run by
\(\SS\)~\cite{I-D.thomson-http-oblivious}. Note that sending such messages via \(\OO\) is compatible with our approach since such messages are encrypted over a TLS connection that is negotiated between each client and \SS. This ensures that we do not introduce any additional trust assumptions to the \tool protocol. Encapsulation and
decapsulation are performed using HPKE, with ciphersuite
\texttt{DHKEM(X25519,HKDF-SHA256)}~\cite{I-D.irtf-cfrg-hpke}.

\point{Other cryptographic machinery}
We implement the \voprf construction detailed by Tyagi et al.~\cite{EPRINT:TCRSTW21}, with~128 bit security. The \voprf is implemented using the
ristretto255 prime-order group  
abstraction.\footnote{\url{https://github.com/dalek-cryptography/curve25519-dalek}}
All hash functions are implemented using SHA-256. All symmetric
encryption is implemented using AES-GCM AEAD with 128-bit keys.

\point{Client measurement sampling} All client inputs are sampled as
256-bit strings from a Zipf power-law distribution with a support of
\(N=10,000\) and parameter \(s = 1.03\). This matches the experimental
choices made in~\cite{SP:BBCGI21}, and captures a large proportion of
applications. This distribution occurs naturally in many network-based
settings~\cite{KleLaw2001} and, as highlighted in~\cite{SP:BBCGI21}, the
chosen parameters are chosen conservatively in that the distribution is
closer to uniform than would typically be expected. In addition, we
measure the costs of \tool in the two cases where clients append either
zero or 256 bytes of auxiliary data to the measurement that they send.

\point{Benchmarking} All benchmarks are run using an AWS EC2
\texttt{c4.8xlarge} instance with~36 vCPUs (3.0 GHz Intel Scalable
Processor) and~60 GiB of memory.

\subsection{Communication Costs}

\begin{figure}[!t]
    \centering
    \includegraphics[scale=0.41]{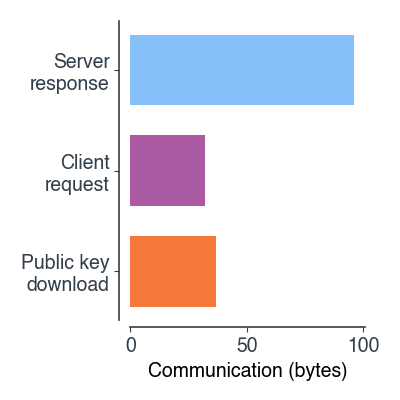}
    \includegraphics[scale=0.41]{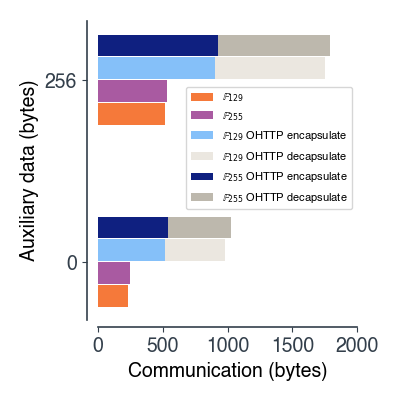}
    \caption{
        \textbf{Left}: Communication with the randomness server during
        the randomness sampling phase of \toolfull.\\
        \textbf{Right}: Communication during the \tool aggregation
        phase. Performance is compared for the two fields
        \(\{\FF_{129},\FF_{255}\}\) used in secret sharing, depending on
        whether OHTTP is utilized, and depending on how much auxiliary
        data is sent (either 0 or 256 bytes) with each measurement.}
    \label{fig:aggregation-communication}
\end{figure}

\point{Randomness server} In \toolfull, the client must request
randomness from the randomness server, which amounts to requesting a
\voprf evaluation on their measurement. We assume that client measurements are  sent daily for seven days (allowing for a daily epoch key rotation of the randomness server key material).
This means that the client must download eight compressed curve points
for the server public key across the seven-day period: amounting to an
amortized download cost of \((8/7)\cdot\texttt{compressed\_ec\_point\_len}\)
bytes per day. The size of a client request is a single compressed
elliptic curve point, and the response is a single curve point, plus two
field scalars for the DLEQ proof. The total amortized per-client
communication costs are given in
Figure~\ref{fig:aggregation-communication}.

\point{Aggregation server} The raw communication costs between clients
and the aggregation server consist of a single encrypted
ciphertext, a secret share, and a 32-byte tag. The size of the share
is dependent on the size of the field that is used. The size of the
ciphertext is dependent on the size of the auxiliary data that is
appended to the client measurement. If client measurements are sent via
the OHTTP proxy, then there are two HTTP requests: one containing an
encapsulated HTTP request to the \emph{proxy resource} and another
corresponding to the decapsulated request to the aggregation server. We
provide per-client communication costs in
Figure~\ref{fig:aggregation-communication}. Note that for constructing
OHTTP requests, we use an encapsulated HTTP request containing the
following information:
\begin{itemize}
    \item HTTP status line: e.g. \texttt{GET /hello.txt HTTP/1.1};
    \item \texttt{User-Agent}, \texttt{Host}, and
    \texttt{Accept-Language} HTTP headers with default values given for each;
    \item \texttt{X-STAR-Message} header containing a base64-encoded \tool
    protocol message (Section~\ref{sec:scheme}).
\end{itemize}

\subsection{Computational Costs}

\begin{figure}[t]
    \centering
    \scalebox{0.9}{\begin{tabular}{
        >{\raggedleft\arraybackslash}p{18mm}%
        >{\raggedleft\arraybackslash}p{18mm}%
        >{\raggedleft\arraybackslash}p{18mm}%
        >{\raggedleft\arraybackslash}p{9mm}%
        >{\raggedleft\arraybackslash}p{9mm}%
    }
        \toprule
        \voprf blind & \voprf final & \voprf verification &
        \multicolumn{2}{p{18mm}}{Aggregation
        message}\\\cmidrule(lr){4-5} & & & \(\FF_{129}\) &
        \(\FF_{255}\)\\\midrule \(0.081\) & \(0.093\) & \(0.301\) & \(0.019\) & \(0.02\)
        \\\bottomrule \end{tabular}} 
        \caption{
        Client runtimes (ms) during the \tool protocol.
        } 
    \label{table:client-computation}
\end{figure}

\begin{figure}[t]
    \centering
    \scalebox{1}{\begin{tabular}{
        >{\raggedleft\arraybackslash}p{18mm}
        >{\raggedleft\arraybackslash}p{18mm}
        >{\raggedleft\arraybackslash}p{18mm}}
        \toprule
        \voprf setup & \voprf evaluation & Proof\newline generation \\\midrule
        \(0.547\) & \(0.662\) & \(0.166\) \\\bottomrule
    \end{tabular}}
    \caption{Randomness server single-threaded runtimes (ms).} 
    \label{table:randomness-server-computation}
\end{figure}

\point{Client message construction}
In Figure~\ref{table:client-computation}, we summarize the various costs
of the cryptographic operations required for each individual client in
\toolfull. Client-side operations are highly efficient: the most expensive are the computation of
two exponentiations in the elliptic curve group that is used. Therefore,
we can reasonably expect that the \tool protocol can be leveraged even
for clients with severely limited computation boundaries. The runtimes
of the randomness server in \toolfull are given in
Figure~\ref{table:randomness-server-computation}.

\begin{figure*}[!t]
    \centering
    \includegraphics[scale=0.4]{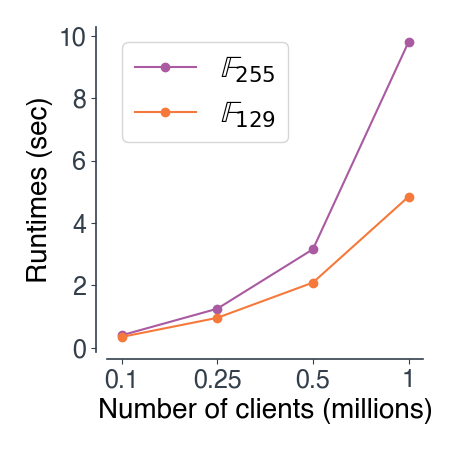}
    \includegraphics[scale=0.4]{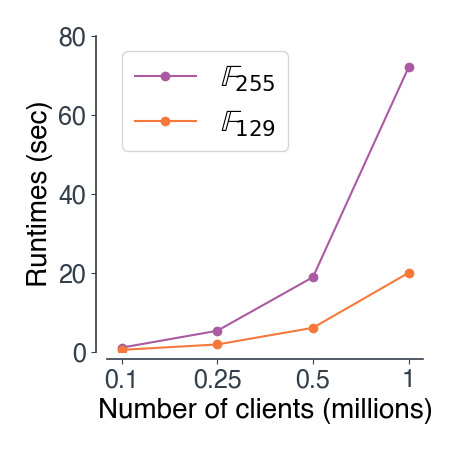}
    \includegraphics[scale=0.4]{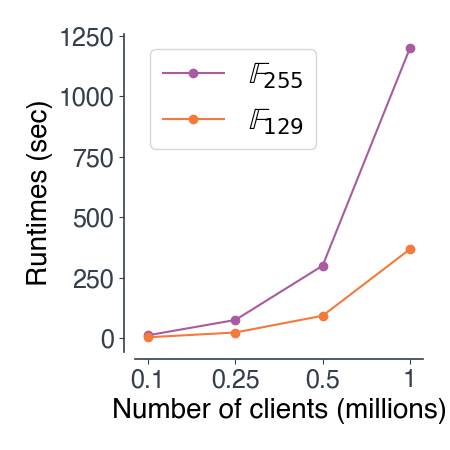}
    \caption{Aggregation server computation runtimes (seconds) based on
        number of clients. Graphs from left to right, corresponding to a
        threshold \(\threshold \in \{0.01\%,0.1\%,1\%\}\) of total
        number of client inputs. Performance is compared for both fields
        \(\{\FF_{129},\FF_{255}\}\).}
    \label{table:server-computation}
\end{figure*}

\point{Aggregation server}
Figure~\ref{table:server-computation} considers the cost of the entire
server aggregation phase for up to~\(1,000,000\) clients, with
\(\threshold\) taken from \(\{0.01\%,0.1\%,1\%\}\) of this number. For
\(1,000,000\) clients with \(\threshold = 0.1\%\), the runtime of the
aggregation server is only~\(20.01s\) using \(\FF_{129}\),
and~\(73.65s\) for \(\FF_{255}\). This clearly indicates that the \tool
protocol is suitable for processing aggregations on very regular
(sub-daily) reporting schedules. Generally, when reducing the underlying
field size (\(\FF_{129}\)) we see runtimes reduce by a factor of
around~\(3\times\). Note that the absolute size of the threshold has a
noticeable impact on the runtime performance, due to the quadratic
overhead of running Lagrange interpolation. This leads to quadratic
growth of runtimes with respect to the threshold.

\begin{figure}[!t]
    \centering
    \centering
    \scalebox{0.87}{\begin{tabular}{
        >{\raggedleft\arraybackslash}p{20mm}%
        >{\raggedleft\arraybackslash}p{20mm}%
        >{\raggedleft\arraybackslash}p{20mm}%
        >{\raggedleft\arraybackslash}p{20mm}%
    }
        \toprule
        Client setup & Server setup & Encapsulate & Decapsulate\\\midrule
        \(0.131\) & \(0.106\) & \(0.002\) & \(0.002\) \\\bottomrule
    \end{tabular}}
    \caption{
        Runtimes (ms) for performing single-threaded HPKE setup,
        encapsulation, and decapsulation at the OHTTP proxy, using the
        \texttt{DHKEM(X25519,~HKDF-SHA256)} ciphersuite.
    }
    \label{table:ohttp-runtimes}
\end{figure}

\point{Oblivious HTTP proxy} Finally, we provide benchmarks in
Figure~\ref{table:ohttp-runtimes} for running HPKE encapsulation and
decapsulation of client messages by the OHTTP proxy. The OHTTP proxy is
required for eliminating client identity leakage.

\subsection{Comparison With Prior Approaches}
We compare \tool directly with the performance
results of the work of Boneh et al.~\cite{SP:BBCGI21}, that devises a
private heavy-hitters protocols from distributed point functions~\cite{EC:GilIsh14}. As shown in Figure~\ref{table:comparison} and mentioned previously, alternative approaches (such as those based on randomised response, MPC, and shuffling) do not provide satisfactory performance or functionality.

To ensure that the leakage profile is similar in both \tool
and~\cite{SP:BBCGI21}, we compare \toolfull performance whilst including
overheads for running the OHTTP proxy. From a functionality perspective,
the protocol of~\cite{SP:BBCGI21} does not allow clients to specify
auxiliary associated data, and thus is not as expressive as the \tool
protocol. For this reason we only consider communication costs
when auxiliary data is not sent. Moreover, \tool requires only a single
aggregation server, while their aggregation phase requires two server
instances. Finally, the client input distribution parameters are
identical. 

\point{Communication}
\toolfull (using \(\FF_{129}\)) requires: \(133\) bytes of
public key data to be downloaded by the client from the randomness
server; \(32\) bytes to be sent by the client to the randomness server;
\(983\) bytes to be sent from the client to the aggregation server (via
the OHTTP proxy), of which only \(464\) bytes is received by the
aggregate server, and \(519\) bytes is received by the OHTTP proxy. This
gives a total \(1148\) bytes per client. The protocol
of~\cite{SP:BBCGI21} requires approximately \(70KB\) of communication
per client. Therefore, overall communication in \toolfull
is~\(\bm{62.4\times}\) \textbf{smaller} than in~\cite{SP:BBCGI21}. Using
\(\FF_{255}\) instead, per-client communication in \toolfull only
increases by 20 bytes.

\point{Runtimes} \toolfull vastly improves on the runtimes
of~\cite{SP:BBCGI21}~---~using \(\FF_{129}\) as the base secret sharing
field, and \(\threshold = 0.1\%\) (the same value used by~\cite{SP:BBCGI21}) of
\(100,000\) clients, \toolfull performs server-side aggregation in
\(0.467\) seconds (and \(1.03\) seconds using \(\FF_{255}\)). Moreover,
times scale reasonably: for~\(500,000\) clients, \toolfull performs
server-side aggregation in~\(6.06s\); for~\(1\) million clients, it
takes \(20s\).\footnote{Using \(\FF_{255}\),
aggregations of data from \(500,000\) and \(1\) million clients take \(18.9s\) and \(72.1s\), respectively.} In contrast, the~\cite{SP:BBCGI21} protocol takes
\(828.1s\) to perform an aggregation of data from \(100,000\) clients,
and \(54\) minutes for~\(400,000\) clients. Thus, the aggregation phase
is \(\bm{1773\times}\) \textbf{faster} in the \toolfull
protocol.

Clients messages take \(0.628ms\) to construct,
including interactions with the randomness server and HPKE
encapsulation. The randomness server operations take \(0.828ms\) per
client input; setup costs occur once and can thus be amortized
across all client messages. The cost of
running the HPKE proxy is \(0.002ms\) per client input. These
times can be distributed across the randomness server key epoch, and requests can be
answered in parallel.

\begin{figure*}[t]
    \centering
    \scalebox{0.8}{\begin{tabular}{%
        >{\Large}l
        >{\Large\raggedleft\arraybackslash}p{39mm}
        >{\Large\raggedleft\arraybackslash}p{35mm}
        >{\Large\raggedleft\arraybackslash}p{35mm}
        >{\Large\raggedleft\arraybackslash}p{35mm}
    }
        \toprule
        Cost & Boneh et al.~\cite{SP:BBCGI21} & \multicolumn{3}{>{\Large\centering}p{105mm}}{\toolfull}\\\cmidrule(lr){2-2}\cmidrule(lr){3-5}
        & Aggregation & Aggregation & \voprf & OHTTP proxy\\\midrule
        { Comms in} & { \(\$0.6193\)} & { \(\$0.00389\)} & { \(\$0.00027\)} & { \(\$0.00435\)} \\
        { Comms out} & { \(\$0.13\)} & { ---} & { \(\$0.00017\)} & { \(\$0.00086\)} \\
        { Computation} & { \(\$0.3659\)} & { \(\$0.0002\)} & { \(\$0.03659\)} & { \(\$0.00009\)} \\\midrule
        { \textbf{Total cost}} & { \(\bm{\$1.1152}\)} & { \(\bm{\$0.00409}\)} & { \(\bm{\$0.03703}\)} & { \(\bm{\$0.0053}\)} \\\bottomrule
    \end{tabular}}

    \caption{Monetary costs associated with running both \toolfull
        and~\cite{SP:BBCGI21}, for aggregating \(100,000\) client
        measurements. All costs include communication from clients and,
        in the case of~\cite{SP:BBCGI21}, communication between
        aggregation servers. Costs for \toolfull include the additional
        costs associated with running the randomness server and OHTTP
        \emph{proxy resource}. All costs are derived from Amazon EC2
        \texttt{c4.8xlarge} costs at time of writing (February 2022). } 
    \label{table:monetary-costs}
\end{figure*}

\point{Financial costs} Finally, taking the costs of running an AWS EC2
\texttt{c4.8xlarge} at the time of writing, it costs \(\$1.591\) per
hour of runtime, plus \(\$0.09\) per GB of data transferred out, and
\(\$0.02\) per GB of data transferred in.\footnote{February 2022} We summarize the monetary
costs for both protocols in Figure~\ref{table:monetary-costs}.
Communication costs are calculated by considering all data transferred
in and out of EC2 instances, and computation costs by considering
computation per hour.\footnote{In~\cite{SP:BBCGI21}, computational is doubled due to the two-server setup.} For \(100,000\) clients, total costs
of running all the components in \toolfull are \(\bm{\$0.00409} +
\bm{\$0.037} + \bm{\$0.0053} = \bm{\$0.04639}\), which is more than
\(\bm{24\times}\) cheaper than the cost of running the Boneh et
al.~\cite{SP:BBCGI21} protocol (\(\bm{\$1.1152}\)). Notice that \tool
remains cheaper than this benchmark even when aggregating data from
\(1,000,000\) clients, costing \(\bm{\$0.4727}\) to run. Since the
monetary costs of running~\cite{SP:BBCGI21} are expected to scale
similarly linearly, we expect that \tool will remain significantly
cheaper beyond \(1,000,000\) clients as well.\footnote{Dominant
financial costs for \tool relate to bandwidth usage, which scale
linearly, rather than aggregation computation time.}

\section{Discussion}
\label{sec:discussion}
%!TEX root=../main.tex

\subsection{\texorpdfstring{Candidate Input Distributions for \toollite}{Entropic Input Distributions for \toollite}}
\label{sec:starlite-dists}
The \toollite protocol must only be used when client inputs
that are \emph{not} eventually revealed are sufficiently entropic;
client inputs that \emph{are} revealed can be drawn from predictable
distributions (Appendix~\ref{app:starlite-sec}). Large
\emph{heavy-tailed} distributions, with correspondingly small thresholds
that ensures the distribution tail has sufficient
min-entropy, appear suitable for ensuring enumeration
attacks are difficult.

It was noted
in~\cite{prochlo} that full URLs form a large, unpredictable search
space. Other wide distributions include the IPv6 address space, which is 64
bits long, and if clients are submitting their own IP addresses these
are likely to be unpredictable and not shared by other clients. Finally,
\tool allows for multiple messages, sampled from independent
distributions, to be concatenated together into a single message. Concatenating enough independently distributed messages can lead
to a distribution that derives enough entropy from each of
the underlying distributions to construct a secure client message.
Finally, the ability of an aggregation server to perform local attacks
can be restricted by using deliberately slower cryptographic algorithms,
as discussed in Section~\ref{sec:scheme}.

We reemphasize that extreme care should be taken when using \toollite,
since making categorical arguments about the entropy present in a
real-world input distribution is very difficult. In most cases, using
\tool is the safest option and comes with very small additional overheads.

\subsection{Limitations}
\label{sec:limitations}
One limitation of \tool is that identity leakage can only be eliminated using
application-layer solutions that anonymize client messages to the
aggregation server (i.e. via an anonymizing proxy). However, note that
some applications (such as those that involve checking for client-side
fraud) may not want to elide such leakage, and thus \tool maintains
flexibility. A further limitation is that \tool cannot provide security
for small message spaces, since this would allow a malicious aggregation
server to enumerate all possible client inputs before it has received
them, via interaction with the randomness server. This limitation is
also possible to exploit in prior systems but with attack complexity
equal to \(n\cdot\threshold\), rather than \(n\) in \tool. Finally, as
is the case for all threshold aggregation systems, \tool remains
vulnerable to Sybil attacks. Preventing such attacks is out-of-scope for
this work, beyond showing that \tool is robust against adversarial
clients to the extent that their only power is in choosing arbitrary
inputs (Theorem~\ref{thm:robustness}, Appendix~\ref{app:cryptographic-guarantees}).

\section{Related Work}
\label{sec:related-work}
We summarize a number of prior approaches that aim to preserve client privacy during threshold aggregation.

\point{Data shuffling} Systems such as Prochlo~\cite{prochlo} construct
a data pipeline for clients to provide measurements whilst maintaining
crowd-based privacy. Clients send their data to an initial server that
strips identifying information and collates measurements into
groupings.\footnote{This process is compatible with adding differential privacy.} Once groupings are large enough, the data is
shuffled and sent to a processing server that can perform general
post-processing. Unfortunately, these pipelines rely on honest execution
of each of the pipeline steps by non-colluding servers, or by trusted
hardware and software enclaves. 
Similar approaches using mix-nets~\cite{Chaum81} and
verifiable shuffling~\cite{CCS:Neff01} provide better security
guarantees, but require increased interactivity to ensure privacy for
thresholds greater than one.

\point{Generic multi-party computation} Generic multi-party computation
(MPC) protocols can be leveraged to
compute threshold aggregation functionality over data from multiple clients~\cite{C:KisSon05,AsiaCCS:BlaAgu12}. In this context, the server
only learns those values which are shared with it over \(\threshold\)
times. Such protocols can be computed directly between clients and
servers using generic two-party computation that ensures malicious
security. Some proposals focus on performing oblivious RAM computations
during client-server
interactions~\cite{FOCS:GarLuOst15,CCS:GKKKMR12,EC:KelYan18,TCC:LuOst13}.
Unfortunately, such protocols remain impractically expensive
for real-world systems~\cite{CCS:DoeShe17}. Moreover, such
schemes require heavily-involved implementations
for instituting the online (and multi-round) communication
and computation patterns.

\point{Outsourced computation} 
Private heavy-hitters protocols aim to provide efficient threshold aggregation
functionality for all data types. A promising,
recent construction explored by Boneh et al.~\cite{SP:BBCGI21} requires
clients to secret-share or \emph{distribute} a point function
(evaluating to \(1\) on their chosen value, and \(0\) elsewhere) between
two aggregation servers~\cite{EC:GilIsh14}. These servers then combine shares of multiple
point functions obliviously and reveal the heavy-hitters among the
client values. Overall, for \(400,000\) clients each holding a 256-bit
string, it takes the two servers \(54\) minutes to compute the
\(\threshold\)-heavy-hitters (where \(\threshold = 0.1\%\) of all
clients) in the dataset, requiring \(70KB\) total communication per
client. This approach leaks all heavy-hitting prefixes,
and more generally all information leaked by the multi-set of honest
client inputs. This information can be restricted by using local
differential privacy.

Outsourcing of said computations had been explored previously in
using \(> 2\) servers, which then interact with
each other to compute the eventual
output~\cite{AsiaCCS:BlaAgu12,C:KisSon05}. Such constructions lead to
computation complexities that are quadratic in the number of
client inputs, and require usage of notably heavier cryptographic
primitives. While more efficient approaches do exist, such as
Prio~\cite{NSDI:CorBon17,EPRINT:AGJO21}, they only allow numerical
inputs, and still incur overheads that are infeasible for
building efficient threshold aggregation systems~\cite{SP:BBCGI21}.

\point{Single-server frameworks for private heavy-hitters}
Randomized response based on local differential privacy (LDP) can
provide private heavy-hitter aggregation that is computed only by a
single
server~\cite{JMLR:BNST20,STOC:BasSmi15,ACMTA:BNS19,CCS:QYYKXR16,PMLR:ZKMSL20}.
The major downside of these approaches is that they do not provide
satisfactory correctness guarantees in all situations (a non-negligible
amount of errors may occur). In particular, when the number of clients
is anything but very large, then the amount of noise introduced is
likely to heavily skew the correctness of the aggregation.\footnote{Conversely, when
the number of clients is very large, the noise that is introduced will be relatively small in comparison to the signal.} Furthermore, since the utility of the system is highly
dependent on the privacy parameter and the number of clients, a system
built upon randomized response requires each operator to make informed
decisions about whether the correctness signal is strong enough for
their application. In addition, solutions based on randomized response
leak a non-negligible amount of information about each client's private
value, since they also include prefix-based leakage similar to the protocol of~\cite{SP:BBCGI21}. We prefer to focus on building a
system that provides perfect correctness and concrete security
guarantees, without having to force implementers to establish security parameterizations or consider correctness guarantees themselves.

Generic approaches for achieving randomized response, such as
systems like RAPPOR~\cite{CCS:ErlPihKor14}, require clients to send a
number of bits that is similar in size to the entire universe of
possible input measurements. As a result, such techniques are infeasible
for situations where this universe is very large.

\point{Secret sharing of client data}
\label{sec:apple}
The \tool construction has similar properties to parts of the secret
sharing approach used by Apple, in their concurrent work to prevent the
spread of Child Sexual Abuse Material (CSAM) on Apple
devices~\cite{AppleCSAM}. Similarities appear in the manner that clients
construct messages to the aggregation server --- using a secret sharing
approach to share media from each of their devices. However, the Apple approach does not extend to a distributed setting,
and only operates across a single client's shares. Our work tackles the
broader question of how clients can non-interactively agree on
compatible secret shares in a distributed system, allowing recovery of
messages that are shared by a threshold number of clients. The wider
system and application are also significantly different.

\section{Conclusion}
\label{sec:conclusion}
%!TEX root=../main.tex

In this work we build \tool: a simple, practical mechanism for threshold
aggregation of client measurements. We intend \tool to enable privacy-protecting,
user-respecting data collection practices that were not practical
or affordable given the existing state-of-the-art. \tool{} is orders of magnitude cheaper,
easier to understand, and easier to implement (in terms of code and trust requirements) than
existing systems. We provide a tested, open source implementation of \tool{}\repourl{}
in rust that can be used in projects today. 
We hope that \tool{} will result in analytics and usage data collection being more private, for more users, and benefiting more analytics frameworks.

% CCS stuff
% \newpage{}

%%%% USENIX STYLE
\includefull{
    \section*{Acknowledgements}
    The authors would like to thank Eric Rescorla, Subodh Iyengar, Ananth Raghunathan, and anonymous reviewers for their helpful feedback on this work.
}{}
    
\bibliographystyle{ACM-Reference-Format}
\bibliography{local,abbrev3,crypto}


\begin{thebibliography}{41}

%%% the Web address.

\ifx \showCODEN    \undefined \def \showCODEN     #1{\unskip}     \fi
\ifx \showDOI      \undefined \def \showDOI       #1{#1}\fi
\ifx \showISBNx    \undefined \def \showISBNx     #1{\unskip}     \fi
\ifx \showISBNxiii \undefined \def \showISBNxiii  #1{\unskip}     \fi
\ifx \showISSN     \undefined \def \showISSN      #1{\unskip}     \fi
\ifx \showLCCN     \undefined \def \showLCCN      #1{\unskip}     \fi
\ifx \shownote     \undefined \def \shownote      #1{#1}          \fi
\ifx \showarticletitle \undefined \def \showarticletitle #1{#1}   \fi
\ifx \showURL      \undefined \def \showURL       {\relax}        \fi
% invisible to TeX
\providecommand\bibfield[2]{#2}
\providecommand\bibinfo[2]{#2}
\providecommand\natexlab[1]{#1}
\providecommand\showeprint[2][]{arXiv:#2}

\bibitem[\protect\citeauthoryear{Addanki, Garbe, Jaffe, Ostrovsky, and
  Polychroniadou}{Addanki et~al\mbox{.}}{2021}]%
        {EPRINT:AGJO21}
\bibfield{author}{\bibinfo{person}{Surya Addanki}, \bibinfo{person}{Kevin
  Garbe}, \bibinfo{person}{Eli Jaffe}, \bibinfo{person}{Rafail Ostrovsky},
  {and} \bibinfo{person}{Antigoni Polychroniadou}.}
  \bibinfo{year}{2021}\natexlab{}.
\newblock \bibinfo{title}{Prio+: Privacy Preserving Aggregate Statistics via
  Boolean Shares}.
\newblock \bibinfo{howpublished}{Cryptology ePrint Archive, Report 2021/576}.
\newblock
\newblock
\shownote{\url{https://eprint.iacr.org/2021/576}.}


\bibitem[\protect\citeauthoryear{Albrecht, Davidson, Deo, and Smart}{Albrecht
  et~al\mbox{.}}{2021}]%
        {PKC:ADDS21}
\bibfield{author}{\bibinfo{person}{Martin~R. Albrecht}, \bibinfo{person}{Alex
  Davidson}, \bibinfo{person}{Amit Deo}, {and} \bibinfo{person}{Nigel~P.
  Smart}.} \bibinfo{year}{2021}\natexlab{}.
\newblock \showarticletitle{Round-Optimal Verifiable Oblivious Pseudorandom
  Functions from Ideal Lattices}. In \bibinfo{booktitle}{\emph{PKC~2021,
  Part~II}} \emph{(\bibinfo{series}{{LNCS}}, Vol.~\bibinfo{volume}{12711})},
  \bibfield{editor}{\bibinfo{person}{Juan Garay}} (Ed.).
  \bibinfo{publisher}{Springer, Heidelberg}, \bibinfo{pages}{261--289}.
\newblock
\urldef\tempurl%
\url{https://doi.org/10.1007/978-3-030-75248-4_10}
\showDOI{\tempurl}


\bibitem[\protect\citeauthoryear{Barnes, Bhargavan, Lipp, and Wood}{Barnes
  et~al\mbox{.}}{2021}]%
        {I-D.irtf-cfrg-hpke}
\bibfield{author}{\bibinfo{person}{Richard Barnes},
  \bibinfo{person}{Karthikeyan Bhargavan}, \bibinfo{person}{Benjamin Lipp},
  {and} \bibinfo{person}{Christopher~A. Wood}.}
  \bibinfo{year}{2021}\natexlab{}.
\newblock \bibinfo{booktitle}{\emph{Hybrid Public Key Encryption}}.
\newblock \bibinfo{type}{Internet-Draft} draft-irtf-cfrg-hpke-12.
  \bibinfo{institution}{IETF Secretariat}.
\newblock
\urldef\tempurl%
\url{https://www.ietf.org/archive/id/draft-irtf-cfrg-hpke-12.txt}
\showURL{%
\tempurl}
\newblock
\shownote{\url{https://www.ietf.org/archive/id/draft-irtf-cfrg-hpke-12.txt}.}


\bibitem[\protect\citeauthoryear{Bassily, Nissim, Stemmer, and
  Thakurta}{Bassily et~al\mbox{.}}{2020}]%
        {JMLR:BNST20}
\bibfield{author}{\bibinfo{person}{Raef Bassily}, \bibinfo{person}{Kobbi
  Nissim}, \bibinfo{person}{Uri Stemmer}, {and} \bibinfo{person}{Abhradeep
  Thakurta}.} \bibinfo{year}{2020}\natexlab{}.
\newblock \showarticletitle{Practical Locally Private Heavy Hitters}.
\newblock \bibinfo{journal}{\emph{Journal of Machine Learning Research}}
  \bibinfo{volume}{21}, \bibinfo{number}{16} (\bibinfo{year}{2020}),
  \bibinfo{pages}{1--42}.
\newblock
\urldef\tempurl%
\url{http://jmlr.org/papers/v21/18-786.html}
\showURL{%
\tempurl}


\bibitem[\protect\citeauthoryear{Bassily and Smith}{Bassily and Smith}{2015}]%
        {STOC:BasSmi15}
\bibfield{author}{\bibinfo{person}{Raef Bassily} {and} \bibinfo{person}{Adam~D.
  Smith}.} \bibinfo{year}{2015}\natexlab{}.
\newblock \showarticletitle{Local, Private, Efficient Protocols for Succinct
  Histograms}. In \bibinfo{booktitle}{\emph{47th ACM STOC}},
  \bibfield{editor}{\bibinfo{person}{Rocco~A. Servedio} {and}
  \bibinfo{person}{Ronitt Rubinfeld}} (Eds.). \bibinfo{publisher}{{ACM} Press},
  \bibinfo{pages}{127--135}.
\newblock
\urldef\tempurl%
\url{https://doi.org/10.1145/2746539.2746632}
\showDOI{\tempurl}


\bibitem[\protect\citeauthoryear{Bellare, Dai, and Rogaway}{Bellare
  et~al\mbox{.}}{2020}]%
        {PoPETS:BelDaiRog20}
\bibfield{author}{\bibinfo{person}{Mihir Bellare}, \bibinfo{person}{Wei Dai},
  {and} \bibinfo{person}{Phillip Rogaway}.} \bibinfo{year}{2020}\natexlab{}.
\newblock \showarticletitle{Reimagining Secret Sharing: Creating a Safer and
  More Versatile Primitive by Adding Authenticity, Correcting Errors, and
  Reducing Randomness Requirements}.
\newblock \bibinfo{journal}{\emph{{PoPETs}}} \bibinfo{volume}{2020},
  \bibinfo{number}{4} (\bibinfo{date}{Oct.} \bibinfo{year}{2020}),
  \bibinfo{pages}{461--490}.
\newblock
\urldef\tempurl%
\url{https://doi.org/10.2478/popets-2020-0082}
\showDOI{\tempurl}


\bibitem[\protect\citeauthoryear{Bhowmick, Boneh, Myers, Talwar, and
  Tarbe}{Bhowmick et~al\mbox{.}}{2021}]%
        {AppleCSAM}
\bibfield{author}{\bibinfo{person}{Abhishek Bhowmick}, \bibinfo{person}{Dan
  Boneh}, \bibinfo{person}{Steve Myers}, \bibinfo{person}{Kunal Talwar}, {and}
  \bibinfo{person}{Karl Tarbe}.} \bibinfo{year}{July 29, 2021}\natexlab{}.
\newblock \bibinfo{title}{The Apple PSI System, Apple Inc.}
\newblock
\newblock
\newblock
\shownote{\url{https://www.apple.com/child-safety/pdf/Apple_PSI_System_Security_Protocol_and_Analysis.pdf}
  (accessed 19 Aug 2021).}


\bibitem[\protect\citeauthoryear{Bittau, Erlingsson, Maniatis, Mironov,
  Raghunathan, Lie, Rudominer, Kode, Tinnes, and Seefeld}{Bittau
  et~al\mbox{.}}{2017}]%
        {prochlo}
\bibfield{author}{\bibinfo{person}{Andrea Bittau}, \bibinfo{person}{\'{U}lfar
  Erlingsson}, \bibinfo{person}{Petros Maniatis}, \bibinfo{person}{Ilya
  Mironov}, \bibinfo{person}{Ananth Raghunathan}, \bibinfo{person}{David Lie},
  \bibinfo{person}{Mitch Rudominer}, \bibinfo{person}{Ushasree Kode},
  \bibinfo{person}{Julien Tinnes}, {and} \bibinfo{person}{Bernhard Seefeld}.}
  \bibinfo{year}{2017}\natexlab{}.
\newblock \showarticletitle{Prochlo: Strong Privacy for Analytics in the
  Crowd}. In \bibinfo{booktitle}{\emph{Proceedings of the 26th Symposium on
  Operating Systems Principles}} (Shanghai, China) \emph{(\bibinfo{series}{SOSP
  '17})}. \bibinfo{publisher}{Association for Computing Machinery},
  \bibinfo{address}{New York, NY, USA}, \bibinfo{pages}{441–459}.
\newblock
\showISBNx{9781450350853}
\urldef\tempurl%
\url{https://doi.org/10.1145/3132747.3132769}
\showDOI{\tempurl}


\bibitem[\protect\citeauthoryear{Blanton and Aguiar}{Blanton and
  Aguiar}{2012}]%
        {AsiaCCS:BlaAgu12}
\bibfield{author}{\bibinfo{person}{Marina Blanton} {and}
  \bibinfo{person}{Everaldo Aguiar}.} \bibinfo{year}{2012}\natexlab{}.
\newblock \showarticletitle{Private and oblivious set and multiset operations}.
  In \bibinfo{booktitle}{\emph{ASIACCS 12}},
  \bibfield{editor}{\bibinfo{person}{Heung~Youl Youm} {and}
  \bibinfo{person}{Yoojae Won}} (Eds.). \bibinfo{publisher}{{ACM} Press},
  \bibinfo{pages}{40--41}.
\newblock


\bibitem[\protect\citeauthoryear{Boneh, Boyle, Corrigan{-}Gibbs, Gilboa, and
  Ishai}{Boneh et~al\mbox{.}}{2021}]%
        {SP:BBCGI21}
\bibfield{author}{\bibinfo{person}{Dan Boneh}, \bibinfo{person}{Elette Boyle},
  \bibinfo{person}{Henry Corrigan{-}Gibbs}, \bibinfo{person}{Niv Gilboa}, {and}
  \bibinfo{person}{Yuval Ishai}.} \bibinfo{year}{2021}\natexlab{}.
\newblock \bibinfo{title}{Lightweight Techniques for Private Heavy Hitters}.
\newblock \bibinfo{howpublished}{IEEE Security \& Privacy}.
\newblock
\newblock
\shownote{\url{https://eprint.iacr.org/2021/017}.}


\bibitem[\protect\citeauthoryear{Bourdrez, Krawczyk, Lewi, and Wood}{Bourdrez
  et~al\mbox{.}}{2021}]%
        {I-D.irtf-cfrg-opaque}
\bibfield{author}{\bibinfo{person}{Daniel Bourdrez}, \bibinfo{person}{Hugo
  Krawczyk}, \bibinfo{person}{Kevin Lewi}, {and}
  \bibinfo{person}{Christopher~A. Wood}.} \bibinfo{year}{2021}\natexlab{}.
\newblock \bibinfo{booktitle}{\emph{The OPAQUE Asymmetric PAKE Protocol}}.
\newblock \bibinfo{type}{Internet-Draft} draft-irtf-cfrg-opaque-07.
  \bibinfo{institution}{IETF Secretariat}.
\newblock
\urldef\tempurl%
\url{https://www.ietf.org/archive/id/draft-irtf-cfrg-opaque-07.txt}
\showURL{%
\tempurl}
\newblock
\shownote{\url{https://www.ietf.org/archive/id/draft-irtf-cfrg-opaque-07.txt}.}


\bibitem[\protect\citeauthoryear{Bun, Nelson, and Stemmer}{Bun
  et~al\mbox{.}}{2019}]%
        {ACMTA:BNS19}
\bibfield{author}{\bibinfo{person}{Mark Bun}, \bibinfo{person}{Jelani Nelson},
  {and} \bibinfo{person}{Uri Stemmer}.} \bibinfo{year}{2019}\natexlab{}.
\newblock \showarticletitle{Heavy Hitters and the Structure of Local Privacy}.
\newblock \bibinfo{journal}{\emph{ACM Trans. Algorithms}} \bibinfo{volume}{15},
  \bibinfo{number}{4}, Article \bibinfo{articleno}{51} (\bibinfo{date}{Oct.}
  \bibinfo{year}{2019}), \bibinfo{numpages}{40}~pages.
\newblock
\showISSN{1549-6325}
\urldef\tempurl%
\url{https://doi.org/10.1145/3344722}
\showDOI{\tempurl}


\bibitem[\protect\citeauthoryear{Chaum}{Chaum}{1981}]%
        {Chaum81}
\bibfield{author}{\bibinfo{person}{David~L. Chaum}.}
  \bibinfo{year}{1981}\natexlab{}.
\newblock \showarticletitle{Untraceable Electronic Mail, Return Addresses, and
  Digital Pseudonyms}.
\newblock \bibinfo{journal}{\emph{Commun. ACM}} \bibinfo{volume}{24},
  \bibinfo{number}{2} (\bibinfo{date}{Feb.} \bibinfo{year}{1981}),
  \bibinfo{pages}{84–90}.
\newblock
\showISSN{0001-0782}
\urldef\tempurl%
\url{https://doi.org/10.1145/358549.358563}
\showDOI{\tempurl}


\bibitem[\protect\citeauthoryear{Corrigan-Gibbs and Boneh}{Corrigan-Gibbs and
  Boneh}{2017}]%
        {NSDI:CorBon17}
\bibfield{author}{\bibinfo{person}{Henry Corrigan-Gibbs} {and}
  \bibinfo{person}{Dan Boneh}.} \bibinfo{year}{2017}\natexlab{}.
\newblock \showarticletitle{Prio: Private, Robust, and Scalable Computation of
  Aggregate Statistics}. In \bibinfo{booktitle}{\emph{Proceedings of the 14th
  USENIX Conference on Networked Systems Design and Implementation}} (Boston,
  MA, USA) \emph{(\bibinfo{series}{NSDI'17})}. \bibinfo{publisher}{USENIX
  Association}, \bibinfo{address}{USA}, \bibinfo{pages}{259–282}.
\newblock
\showISBNx{9781931971379}


\bibitem[\protect\citeauthoryear{Davidson, Faz-Hernandez, Sullivan, and
  Wood}{Davidson et~al\mbox{.}}{2021}]%
        {I-D.irtf-cfrg-voprf}
\bibfield{author}{\bibinfo{person}{Alex Davidson}, \bibinfo{person}{Armando
  Faz-Hernandez}, \bibinfo{person}{Nick Sullivan}, {and}
  \bibinfo{person}{Christopher~A. Wood}.} \bibinfo{year}{2021}\natexlab{}.
\newblock \bibinfo{booktitle}{\emph{Oblivious Pseudorandom Functions (OPRFs)
  using Prime-Order Groups}}.
\newblock \bibinfo{type}{Internet-Draft} draft-irtf-cfrg-voprf-06.
  \bibinfo{institution}{IETF Secretariat}.
\newblock
\urldef\tempurl%
\url{https://www.ietf.org/archive/id/draft-irtf-cfrg-voprf-06.txt}
\showURL{%
\tempurl}
\newblock
\shownote{\url{https://www.ietf.org/archive/id/draft-irtf-cfrg-voprf-06.txt}.}


\bibitem[\protect\citeauthoryear{Davidson, Goldberg, Sullivan, Tankersley, and
  Valsorda}{Davidson et~al\mbox{.}}{2018}]%
        {PoPETS:DGSTV18}
\bibfield{author}{\bibinfo{person}{Alex Davidson}, \bibinfo{person}{Ian
  Goldberg}, \bibinfo{person}{Nick Sullivan}, \bibinfo{person}{George
  Tankersley}, {and} \bibinfo{person}{Filippo Valsorda}.}
  \bibinfo{year}{2018}\natexlab{}.
\newblock \showarticletitle{Privacy Pass: Bypassing Internet Challenges
  Anonymously}.
\newblock \bibinfo{journal}{\emph{{PoPETs}}} \bibinfo{volume}{2018},
  \bibinfo{number}{3} (\bibinfo{date}{July} \bibinfo{year}{2018}),
  \bibinfo{pages}{164--180}.
\newblock
\urldef\tempurl%
\url{https://doi.org/10.1515/popets-2018-0026}
\showDOI{\tempurl}


\bibitem[\protect\citeauthoryear{Doerner and {shelat}}{Doerner and
  {shelat}}{2017}]%
        {CCS:DoeShe17}
\bibfield{author}{\bibinfo{person}{Jack Doerner} {and} \bibinfo{person}{{abhi}
  {shelat}}.} \bibinfo{year}{2017}\natexlab{}.
\newblock \showarticletitle{Scaling {ORAM} for Secure Computation}. In
  \bibinfo{booktitle}{\emph{ACM CCS 2017}},
  \bibfield{editor}{\bibinfo{person}{Bhavani~M. Thuraisingham},
  \bibinfo{person}{David Evans}, \bibinfo{person}{Tal Malkin}, {and}
  \bibinfo{person}{Dongyan Xu}} (Eds.). \bibinfo{publisher}{{ACM} Press},
  \bibinfo{pages}{523--535}.
\newblock
\urldef\tempurl%
\url{https://doi.org/10.1145/3133956.3133967}
\showDOI{\tempurl}


\bibitem[\protect\citeauthoryear{Douceur}{Douceur}{2002}]%
        {IPTPS:Douceur01}
\bibfield{author}{\bibinfo{person}{John~R. Douceur}.}
  \bibinfo{year}{2002}\natexlab{}.
\newblock \showarticletitle{The Sybil Attack}. In
  \bibinfo{booktitle}{\emph{IPTPS '01: Revised Papers from the First
  International Workshop on Peer-to-Peer Systems}}.
  \bibinfo{publisher}{Springer-Verlag}, \bibinfo{address}{London, UK},
  \bibinfo{pages}{251--260}.
\newblock
\showISBNx{3540441794}
\urldef\tempurl%
\url{http://portal.acm.org/citation.cfm?id=687813}
\showURL{%
\tempurl}


\bibitem[\protect\citeauthoryear{Erlingsson, Pihur, and Korolova}{Erlingsson
  et~al\mbox{.}}{2014}]%
        {CCS:ErlPihKor14}
\bibfield{author}{\bibinfo{person}{{\'U}lfar Erlingsson},
  \bibinfo{person}{Vasyl Pihur}, {and} \bibinfo{person}{Aleksandra Korolova}.}
  \bibinfo{year}{2014}\natexlab{}.
\newblock \showarticletitle{{RAPPOR}: Randomized Aggregatable
  Privacy-Preserving Ordinal Response}. In \bibinfo{booktitle}{\emph{ACM CCS
  2014}}, \bibfield{editor}{\bibinfo{person}{Gail-Joon Ahn},
  \bibinfo{person}{Moti Yung}, {and} \bibinfo{person}{Ninghui Li}} (Eds.).
  \bibinfo{publisher}{{ACM} Press}, \bibinfo{pages}{1054--1067}.
\newblock
\urldef\tempurl%
\url{https://doi.org/10.1145/2660267.2660348}
\showDOI{\tempurl}


\bibitem[\protect\citeauthoryear{Freedman, Ishai, Pinkas, and
  Reingold}{Freedman et~al\mbox{.}}{2005}]%
        {TCC:FIPR05}
\bibfield{author}{\bibinfo{person}{Michael~J. Freedman}, \bibinfo{person}{Yuval
  Ishai}, \bibinfo{person}{Benny Pinkas}, {and} \bibinfo{person}{Omer
  Reingold}.} \bibinfo{year}{2005}\natexlab{}.
\newblock \showarticletitle{Keyword Search and Oblivious Pseudorandom
  Functions}. In \bibinfo{booktitle}{\emph{TCC~2005}}
  \emph{(\bibinfo{series}{{LNCS}}, Vol.~\bibinfo{volume}{3378})},
  \bibfield{editor}{\bibinfo{person}{Joe Kilian}} (Ed.).
  \bibinfo{publisher}{Springer, Heidelberg}, \bibinfo{pages}{303--324}.
\newblock
\urldef\tempurl%
\url{https://doi.org/10.1007/978-3-540-30576-7_17}
\showDOI{\tempurl}


\bibitem[\protect\citeauthoryear{Garg, Lu, and Ostrovsky}{Garg
  et~al\mbox{.}}{2015}]%
        {FOCS:GarLuOst15}
\bibfield{author}{\bibinfo{person}{Sanjam Garg}, \bibinfo{person}{Steve Lu},
  {and} \bibinfo{person}{Rafail Ostrovsky}.} \bibinfo{year}{2015}\natexlab{}.
\newblock \showarticletitle{Black-Box Garbled {RAM}}. In
  \bibinfo{booktitle}{\emph{56th FOCS}},
  \bibfield{editor}{\bibinfo{person}{Venkatesan Guruswami}} (Ed.).
  \bibinfo{publisher}{{IEEE} Computer Society Press},
  \bibinfo{pages}{210--229}.
\newblock
\urldef\tempurl%
\url{https://doi.org/10.1109/FOCS.2015.22}
\showDOI{\tempurl}


\bibitem[\protect\citeauthoryear{Gilboa and Ishai}{Gilboa and Ishai}{2014}]%
        {EC:GilIsh14}
\bibfield{author}{\bibinfo{person}{Niv Gilboa} {and} \bibinfo{person}{Yuval
  Ishai}.} \bibinfo{year}{2014}\natexlab{}.
\newblock \showarticletitle{Distributed Point Functions and Their
  Applications}. In \bibinfo{booktitle}{\emph{EUROCRYPT~2014}}
  \emph{(\bibinfo{series}{{LNCS}}, Vol.~\bibinfo{volume}{8441})},
  \bibfield{editor}{\bibinfo{person}{Phong~Q. Nguyen} {and}
  \bibinfo{person}{Elisabeth Oswald}} (Eds.). \bibinfo{publisher}{Springer,
  Heidelberg}, \bibinfo{pages}{640--658}.
\newblock
\urldef\tempurl%
\url{https://doi.org/10.1007/978-3-642-55220-5_35}
\showDOI{\tempurl}


\bibitem[\protect\citeauthoryear{Gordon, Katz, Kolesnikov, Krell, Malkin,
  Raykova, and Vahlis}{Gordon et~al\mbox{.}}{2012}]%
        {CCS:GKKKMR12}
\bibfield{author}{\bibinfo{person}{S.~Dov Gordon}, \bibinfo{person}{Jonathan
  Katz}, \bibinfo{person}{Vladimir Kolesnikov}, \bibinfo{person}{Fernando
  Krell}, \bibinfo{person}{Tal Malkin}, \bibinfo{person}{Mariana Raykova},
  {and} \bibinfo{person}{Yevgeniy Vahlis}.} \bibinfo{year}{2012}\natexlab{}.
\newblock \showarticletitle{Secure two-party computation in sublinear
  (amortized) time}. In \bibinfo{booktitle}{\emph{ACM CCS 2012}},
  \bibfield{editor}{\bibinfo{person}{Ting Yu}, \bibinfo{person}{George
  Danezis}, {and} \bibinfo{person}{Virgil~D. Gligor}} (Eds.).
  \bibinfo{publisher}{{ACM} Press}, \bibinfo{pages}{513--524}.
\newblock
\urldef\tempurl%
\url{https://doi.org/10.1145/2382196.2382251}
\showDOI{\tempurl}


\bibitem[\protect\citeauthoryear{Huang, Iyengar, Jeyaraman, Kushwah, Chen-Kuei,
  Luo, Mohassel, Raghunathan, Shaikh, Sung, and Zhang}{Huang
  et~al\mbox{.}}{2021}]%
        {DITFB}
\bibfield{author}{\bibinfo{person}{Sharon Huang}, \bibinfo{person}{Subodh
  Iyengar}, \bibinfo{person}{Sundar Jeyaraman}, \bibinfo{person}{Shiv Kushwah},
  \bibinfo{person}{Chen-Kuei}, \bibinfo{person}{Lee~Zutian Luo},
  \bibinfo{person}{Payman Mohassel}, \bibinfo{person}{Ananth Raghunathan},
  \bibinfo{person}{Shaahid Shaikh}, \bibinfo{person}{Yen-Chieh Sung}, {and}
  \bibinfo{person}{Albert Zhang}.} \bibinfo{year}{2021}\natexlab{}.
\newblock \showarticletitle{DIT: De-Identified Authenticated Telemetry at
  Scale}.
\newblock  (\bibinfo{year}{2021}).
\newblock
\newblock
\shownote{\url{https://tinyurl.com/yxt7u2ss}.}


\bibitem[\protect\citeauthoryear{Jarecki, Kiayias, and Krawczyk}{Jarecki
  et~al\mbox{.}}{2014}]%
        {AC:JarKiaKra14}
\bibfield{author}{\bibinfo{person}{Stanislaw Jarecki}, \bibinfo{person}{Aggelos
  Kiayias}, {and} \bibinfo{person}{Hugo Krawczyk}.}
  \bibinfo{year}{2014}\natexlab{}.
\newblock \showarticletitle{Round-Optimal Password-Protected Secret Sharing and
  {T}-{PAKE} in the Password-Only Model}. In
  \bibinfo{booktitle}{\emph{ASIACRYPT~2014, Part~II}}
  \emph{(\bibinfo{series}{{LNCS}}, Vol.~\bibinfo{volume}{8874})},
  \bibfield{editor}{\bibinfo{person}{Palash Sarkar} {and}
  \bibinfo{person}{Tetsu Iwata}} (Eds.). \bibinfo{publisher}{Springer,
  Heidelberg}, \bibinfo{pages}{233--253}.
\newblock
\urldef\tempurl%
\url{https://doi.org/10.1007/978-3-662-45608-8_13}
\showDOI{\tempurl}


\bibitem[\protect\citeauthoryear{Kaliski}{Kaliski}{2000}]%
        {RFC2898}
\bibfield{author}{\bibinfo{person}{B. Kaliski}.}
  \bibinfo{year}{2000}\natexlab{}.
\newblock \bibinfo{booktitle}{\emph{PKCS \#5: Password-Based Cryptography
  Specification Version 2.0}}.
\newblock \bibinfo{type}{RFC} 2898. \bibinfo{institution}{RFC Editor}.
\newblock
\showISSN{2070-1721}
\urldef\tempurl%
\url{http://www.rfc-editor.org/rfc/rfc2898.txt}
\showURL{%
\tempurl}
\newblock
\shownote{\url{http://www.rfc-editor.org/rfc/rfc2898.txt}.}


\bibitem[\protect\citeauthoryear{Keller and Yanai}{Keller and Yanai}{2018}]%
        {EC:KelYan18}
\bibfield{author}{\bibinfo{person}{Marcel Keller} {and}
  \bibinfo{person}{Avishay Yanai}.} \bibinfo{year}{2018}\natexlab{}.
\newblock \showarticletitle{Efficient Maliciously Secure Multiparty Computation
  for {RAM}}. In \bibinfo{booktitle}{\emph{EUROCRYPT~2018, Part~III}}
  \emph{(\bibinfo{series}{{LNCS}}, Vol.~\bibinfo{volume}{10822})},
  \bibfield{editor}{\bibinfo{person}{Jesper~Buus Nielsen} {and}
  \bibinfo{person}{Vincent Rijmen}} (Eds.). \bibinfo{publisher}{Springer,
  Heidelberg}, \bibinfo{pages}{91--124}.
\newblock
\urldef\tempurl%
\url{https://doi.org/10.1007/978-3-319-78372-7_4}
\showDOI{\tempurl}


\bibitem[\protect\citeauthoryear{Kissner and Song}{Kissner and Song}{2005}]%
        {C:KisSon05}
\bibfield{author}{\bibinfo{person}{Lea Kissner} {and}
  \bibinfo{person}{Dawn~Xiaodong Song}.} \bibinfo{year}{2005}\natexlab{}.
\newblock \showarticletitle{Privacy-Preserving Set Operations}. In
  \bibinfo{booktitle}{\emph{CRYPTO~2005}} \emph{(\bibinfo{series}{{LNCS}},
  Vol.~\bibinfo{volume}{3621})}, \bibfield{editor}{\bibinfo{person}{Victor
  Shoup}} (Ed.). \bibinfo{publisher}{Springer, Heidelberg},
  \bibinfo{pages}{241--257}.
\newblock
\urldef\tempurl%
\url{https://doi.org/10.1007/11535218_15}
\showDOI{\tempurl}


\bibitem[\protect\citeauthoryear{Kleinberg and Lawrence}{Kleinberg and
  Lawrence}{2001}]%
        {KleLaw2001}
\bibfield{author}{\bibinfo{person}{Jon Kleinberg} {and} \bibinfo{person}{Steve
  Lawrence}.} \bibinfo{year}{2001}\natexlab{}.
\newblock \showarticletitle{The Structure of the Web}.
\newblock \bibinfo{journal}{\emph{Science}} \bibinfo{volume}{294},
  \bibinfo{number}{5548} (\bibinfo{year}{2001}), \bibinfo{pages}{1849--1850}.
\newblock
\showISSN{0036-8075}
\urldef\tempurl%
\url{https://doi.org/10.1126/science.1067014}
\showDOI{\tempurl}
\showeprint{https://science.sciencemag.org/content/294/5548/1849.full.pdf}


\bibitem[\protect\citeauthoryear{Kreuter, Lepoint, Orr{\`u}, and
  Raykova}{Kreuter et~al\mbox{.}}{2020}]%
        {C:KLOR20}
\bibfield{author}{\bibinfo{person}{Ben Kreuter}, \bibinfo{person}{Tancr{\`e}de
  Lepoint}, \bibinfo{person}{Michele Orr{\`u}}, {and} \bibinfo{person}{Mariana
  Raykova}.} \bibinfo{year}{2020}\natexlab{}.
\newblock \showarticletitle{Anonymous Tokens with Private Metadata Bit}. In
  \bibinfo{booktitle}{\emph{CRYPTO~2020, Part~I}}
  \emph{(\bibinfo{series}{{LNCS}}, Vol.~\bibinfo{volume}{12170})},
  \bibfield{editor}{\bibinfo{person}{Daniele Micciancio} {and}
  \bibinfo{person}{Thomas Ristenpart}} (Eds.). \bibinfo{publisher}{Springer,
  Heidelberg}, \bibinfo{pages}{308--336}.
\newblock
\urldef\tempurl%
\url{https://doi.org/10.1007/978-3-030-56784-2_11}
\showDOI{\tempurl}


\bibitem[\protect\citeauthoryear{Lu and Ostrovsky}{Lu and Ostrovsky}{2013}]%
        {TCC:LuOst13}
\bibfield{author}{\bibinfo{person}{Steve Lu} {and} \bibinfo{person}{Rafail
  Ostrovsky}.} \bibinfo{year}{2013}\natexlab{}.
\newblock \showarticletitle{Distributed Oblivious {RAM} for Secure Two-Party
  Computation}. In \bibinfo{booktitle}{\emph{TCC~2013}}
  \emph{(\bibinfo{series}{{LNCS}}, Vol.~\bibinfo{volume}{7785})},
  \bibfield{editor}{\bibinfo{person}{Amit Sahai}} (Ed.).
  \bibinfo{publisher}{Springer, Heidelberg}, \bibinfo{pages}{377--396}.
\newblock
\urldef\tempurl%
\url{https://doi.org/10.1007/978-3-642-36594-2_22}
\showDOI{\tempurl}


\bibitem[\protect\citeauthoryear{Neff}{Neff}{2001}]%
        {CCS:Neff01}
\bibfield{author}{\bibinfo{person}{C.~Andrew Neff}.}
  \bibinfo{year}{2001}\natexlab{}.
\newblock \showarticletitle{A Verifiable Secret Shuffle and Its Application to
  e-Voting}. In \bibinfo{booktitle}{\emph{ACM CCS 2001}},
  \bibfield{editor}{\bibinfo{person}{Michael~K. Reiter} {and}
  \bibinfo{person}{Pierangela Samarati}} (Eds.). \bibinfo{publisher}{{ACM}
  Press}, \bibinfo{pages}{116--125}.
\newblock
\urldef\tempurl%
\url{https://doi.org/10.1145/501983.502000}
\showDOI{\tempurl}


\bibitem[\protect\citeauthoryear{Nilsson, Bideh, and Brorsson}{Nilsson
  et~al\mbox{.}}{2020}]%
        {arxiv:NilBidBro20}
\bibfield{author}{\bibinfo{person}{Alexander Nilsson},
  \bibinfo{person}{Pegah~Nikbakht Bideh}, {and} \bibinfo{person}{Joakim
  Brorsson}.} \bibinfo{year}{2020}\natexlab{}.
\newblock \showarticletitle{A survey of published attacks on Intel SGX}.
\newblock \bibinfo{journal}{\emph{arXiv preprint arXiv:2006.13598}}
  (\bibinfo{year}{2020}).
\newblock


\bibitem[\protect\citeauthoryear{Percival and Josefsson}{Percival and
  Josefsson}{2016}]%
        {RFC7914}
\bibfield{author}{\bibinfo{person}{C. Percival} {and} \bibinfo{person}{S.
  Josefsson}.} \bibinfo{year}{2016}\natexlab{}.
\newblock \bibinfo{booktitle}{\emph{The scrypt Password-Based Key Derivation
  Function}}.
\newblock \bibinfo{type}{RFC} 7914. \bibinfo{institution}{RFC Editor}.
\newblock
\showISSN{2070-1721}


\bibitem[\protect\citeauthoryear{Qin, Yang, Yu, Khalil, Xiao, and Ren}{Qin
  et~al\mbox{.}}{2016}]%
        {CCS:QYYKXR16}
\bibfield{author}{\bibinfo{person}{Zhan Qin}, \bibinfo{person}{Yin Yang},
  \bibinfo{person}{Ting Yu}, \bibinfo{person}{Issa Khalil},
  \bibinfo{person}{Xiaokui Xiao}, {and} \bibinfo{person}{Kui Ren}.}
  \bibinfo{year}{2016}\natexlab{}.
\newblock \showarticletitle{Heavy Hitter Estimation over Set-Valued Data with
  Local Differential Privacy}. In \bibinfo{booktitle}{\emph{ACM CCS 2016}},
  \bibfield{editor}{\bibinfo{person}{Edgar~R. Weippl}, \bibinfo{person}{Stefan
  Katzenbeisser}, \bibinfo{person}{Christopher Kruegel},
  \bibinfo{person}{Andrew~C. Myers}, {and} \bibinfo{person}{Shai Halevi}}
  (Eds.). \bibinfo{publisher}{{ACM} Press}, \bibinfo{pages}{192--203}.
\newblock
\urldef\tempurl%
\url{https://doi.org/10.1145/2976749.2978409}
\showDOI{\tempurl}


\bibitem[\protect\citeauthoryear{Rescorla}{Rescorla}{2021}]%
        {ietf-ppm}
\bibfield{author}{\bibinfo{person}{Eric Rescorla}.}
  \bibinfo{year}{2021}\natexlab{}.
\newblock  (\bibinfo{year}{2021}).
\newblock
\newblock
\shownote{\texttt{bofreq-privacy-preserving-measurement-06}
  \url{https://datatracker.ietf.org/doc/bofreq-privacy-preserving-measurement/}.}


\bibitem[\protect\citeauthoryear{Servan-Schreiber, Hogan, and
  Devadas}{Servan-Schreiber et~al\mbox{.}}{2021}]%
        {EPRINT:SerHogDev21}
\bibfield{author}{\bibinfo{person}{Sacha Servan-Schreiber},
  \bibinfo{person}{Kyle Hogan}, {and} \bibinfo{person}{Srinivas Devadas}.}
  \bibinfo{year}{2021}\natexlab{}.
\newblock \bibinfo{title}{AdVeil: A Private Targeted-Advertising Ecosystem}.
\newblock \bibinfo{howpublished}{Cryptology ePrint Archive, Report 2021/1032}.
\newblock
\newblock
\shownote{\url{https://ia.cr/2021/1032}.}


\bibitem[\protect\citeauthoryear{Sweeney}{Sweeney}{2002}]%
        {IJUFKBS:Sweeney02}
\bibfield{author}{\bibinfo{person}{Latanya Sweeney}.}
  \bibinfo{year}{2002}\natexlab{}.
\newblock \showarticletitle{k-Anonymity: {A} Model for Protecting Privacy}.
\newblock \bibinfo{journal}{\emph{Int. J. Uncertain. Fuzziness Knowl. Based
  Syst.}} \bibinfo{volume}{10}, \bibinfo{number}{5} (\bibinfo{year}{2002}),
  \bibinfo{pages}{557--570}.
\newblock
\urldef\tempurl%
\url{https://doi.org/10.1142/S0218488502001648}
\showDOI{\tempurl}


\bibitem[\protect\citeauthoryear{Thomson and Wood}{Thomson and Wood}{2021}]%
        {I-D.thomson-http-oblivious}
\bibfield{author}{\bibinfo{person}{Martin Thomson} {and}
  \bibinfo{person}{Christopher~A. Wood}.} \bibinfo{year}{2021}\natexlab{}.
\newblock \bibinfo{booktitle}{\emph{Oblivious HTTP}}.
\newblock \bibinfo{type}{Internet-Draft} draft-thomson-ohai-ohttp-00.
  \bibinfo{institution}{IETF Secretariat}.
\newblock
\urldef\tempurl%
\url{https://www.ietf.org/archive/id/draft-thomson-ohai-ohttp-00.txt}
\showURL{%
\tempurl}
\newblock
\shownote{\url{https://www.ietf.org/archive/id/draft-thomson-ohai-ohttp-00.txt}.}


\bibitem[\protect\citeauthoryear{Tyagi, Celi, Ristenpart, Sullivan, Tessaro,
  and Wood}{Tyagi et~al\mbox{.}}{2021}]%
        {EPRINT:TCRSTW21}
\bibfield{author}{\bibinfo{person}{Nirvan Tyagi}, \bibinfo{person}{Sofia Celi},
  \bibinfo{person}{Thomas Ristenpart}, \bibinfo{person}{Nick Sullivan},
  \bibinfo{person}{Stefano Tessaro}, {and} \bibinfo{person}{Christopher~A.
  Wood}.} \bibinfo{year}{2021}\natexlab{}.
\newblock \bibinfo{title}{A Fast and Simple Partially Oblivious PRF, with
  Applications}.
\newblock \bibinfo{howpublished}{Cryptology ePrint Archive, Report 2021/864}.
\newblock
\newblock
\shownote{\url{https://eprint.iacr.org/2021/864}.}


\bibitem[\protect\citeauthoryear{Zhu, Kairouz, McMahan, Sun, and Li}{Zhu
  et~al\mbox{.}}{2020}]%
        {PMLR:ZKMSL20}
\bibfield{author}{\bibinfo{person}{Wennan Zhu}, \bibinfo{person}{Peter
  Kairouz}, \bibinfo{person}{Brendan McMahan}, \bibinfo{person}{Haicheng Sun},
  {and} \bibinfo{person}{Wei Li}.} \bibinfo{year}{2020}\natexlab{}.
\newblock \showarticletitle{Federated Heavy Hitters Discovery with Differential
  Privacy}. In \bibinfo{booktitle}{\emph{Proceedings of the Twenty Third
  International Conference on Artificial Intelligence and Statistics}}
  \emph{(\bibinfo{series}{Proceedings of Machine Learning Research},
  Vol.~\bibinfo{volume}{108})}, \bibfield{editor}{\bibinfo{person}{Silvia
  Chiappa} {and} \bibinfo{person}{Roberto Calandra}} (Eds.).
  \bibinfo{publisher}{PMLR}, \bibinfo{pages}{3837--3847}.
\newblock
\urldef\tempurl%
\url{http://proceedings.mlr.press/v108/zhu20a.html}
\showURL{%
\tempurl}


\end{thebibliography}

%%%% USENIX STYLE
\appendix

\section{Cryptographic guarantees}
\label{app:cryptographic-guarantees}
In the following section, we will assume the presence of each of the cryptographic primitives specified in Section~\ref{sec:crypto}. We will use the notation that was detailed in Section~\ref{sec:notation}. The following establishes the correctness and security of the \tool protocol (denoted \P) and the \toollite protocol (denoted \(\widetilde{\P}\)), with respect to the ideal functionality laid out in Section~\ref{sec:security-model}.

\subsection{Correctness}
\label{app:proof-correctness}
We first state the correctness guarantee of the \tool protocol.
\begin{theorem}{\emph{(Correctness)}}
    The protocol \(\P\) (similarly \(\widetilde{\P}\)) is correct with all but negligible probability.
    \label{thm:correctness}
\end{theorem}
\begin{proof}
    The correctness of \tool follows from the fact that \SS recovers a
    symmetric key \(K_{\E_\iota}\) for every subset \(\E_\iota \in \Y\) of
    compatible client shares of size greater than \(\threshold\). In these instances, the server uses
    the value \(t_{\E_\iota}\) to check which shares correspond to each other.
    It then uses the \(\revealnoinput\) procedure to reveal \(r_{{\E_\iota},1}\) and
    derive \(K_{\E_\iota}\). Once it learns \(K_{\E_\iota}\), \SS is able to
    recover \((x_{j},\aux_{j})\) by decrypting each client message corresponding to \(\E_\iota\).
    
    As mentioned in Section~\ref{sec:adss}, this requires that the underlying
    field \(\FF_p\) that secret shares are generated within is created with
    prime order \(p\) large enough. This ensures that randomly sampling shares from this
    field is unlikely to lead to collisions. If a collision
    occurs and the total number of \emph{different} shares is \(<
    \threshold\), then the recovery operation will not succeed.
\end{proof}

\subsection{Security}
\label{app:proof-security}
We prove the security of \tool against a malicious adversary, that is allowed to operate in one of the following manners: corrupting the aggregation server and a set of clients together; corrupting the randomness server and a set of clients together; and corrupting only a set of clients. We show that the \tool protocol maintains client privacy (up to leakage specified by \leakagenoinput) in the case where either server is corrupted. Furthermore, the computation is shown to be robust against an adversary that controls only a set of clients, and attempts to alter the protocol output. The security proofs for \(\P\) are given in Theorems~\ref{thm:privacymalaggserverp2},~\ref{thm:privacymalrandserver}, and~\ref{thm:robustness}. Throughout, we will use \(\F_\P\) to refer to the ideal functionality for the \tool protocol, and \(\F_\voprfschemeshort\) to refer to the ideal functionality for the VOPRF protocol.

\point{Random oracle model usage in \voprf} While the \(\P\)
protocol itself does not include any explicit usage of random oracles,
we require that the internal \voprf protocol uses a random oracle \RO
in the final evaluation of the PRF value. This allows the simulation
to learn adversarial inputs from queries during the protocol execution.
Specifically, we require that the \voprf scheme produces
outputs of the form \(\RO(x,f(\msk,x))\). 
In other words, the ideal functionality of \voprfschemeshort provides an initial output, that the client then finalizes using the random oracle query.
This is actually slightly weaker than a standard VOPRF, but many well-known OPRF primitives adhere to this security model~\cite{EPRINT:TCRSTW21,AC:JarKiaKra14,PoPETS:DGSTV18,C:KLOR20}. 

\point{Security theorems for \tool} We now detail the various theorems that prove
the security of \(\P\).

\begin{theorem}{\emph{(Malicious aggregation
    server)}} The protocol \(\P\) is secure against any
    \(\adv\) that corrupts \SS and some subset \(\cdv_\adv \subset \cdv\) of
    all clients, assuming a secure VOPRF protocol \voprfschemeshort, the \(\indcpa\) security of \(\skeschemeshort\), and the privacy of \(\adssschemeshort{\threshold}{n}\).
    \label{thm:privacymalaggserverp2}
\end{theorem}
\begin{proof}
    We construct our PPT simulator as follows.
\begin{itemize}
    \item \(\sdv\) runs \(\pparams \leftarrow
    \voprfschemeshort.\poprfsetup{\secparam}\) and \((\msk',\mpk') \leftarrow
    \voprfschemeshort.\mskgen{\pparams}\) and sends \(\pparams\)
    to \adv.
    \item \(\sdv\) handles queries made by \(\adv\) to
    \(\voprfschemeshort\) by interacting with the ideal functionality \(\F_{\voprfschemeshort}\).
    \item When \(\sdv\) receives queries \((x,y)\) to the
    random oracle \voprfschemeshort.\RO, it first checks that \(y
    = f(\msk',x)\). If this equality holds, it either returns
    \(\RO[x]\), or samples \(z \sample \bin^{3\omega}\), sets
    \(\RO[x] = z_x\), and then returns \(z_x\). If the inequality
    does not hold, it returns a randomly sampled value.
    \item When \(\sdv\) receives \((\threshold,(c_i,s_i,t_i)_{i \in
    \cdv_\adv})\) from the adversary, it sends all
    inputs \(\X_\adv\) that it received \RO queries for, with
    the set \(\aux_\adv = \{\perp\}_{\iota \in |\X_\adv|}\) and
    \(\threshold\), to \(\F_{\P}\). It
    receives \(\Y\) as output from \(\F_{\P}\), and \(\leakage{\X}\).
    \item Let \(\N\) be the collection of subsets of all
    indices returned by \(\leakage{\X}\), let \(\N_\adv
    \subset \N\) denote all subsets
    that contain inputs taken from \(\X_\adv\), and let \(\Z = \emptyset\). 
    \item For each \((x_j,\aux_j) \in \Y\):
    \begin{itemize}
        \item If \(z_{x_j} = \RO[x_j]\) is not empty, then let:
        \begin{align}
            \begin{split}
                K_{x_j} &\leftarrow \derive{z_{x_j}[1],\secparam};\\
                c_{x_j} &\leftarrow \skeschemeshort.\enc(K_{x_j},x\|\aux_j);\\
                s_{x_j} &\leftarrow \adssschemeshort{\threshold}{n}.\share{z_{x_j}[1];z_{x_j}[2]};\\
                t_{x_j} &\leftarrow z_{x_j}[3].
            \end{split}
            \label{eq:random_oracle}
        \end{align}
        Else, sample \(z_{x_j} \sample \bin^{3\omega}\) and construct
        \((c_{x_j},s_{x_j},t_{x_j})\) as in
        Equation~\eqref{eq:simulate}.
        \item Let \(\Z[j] = (c_{x_j},s_{x_j},t_{x_j})\).
    \end{itemize}
    \item For each \(j\) where \(\aux_j = \perp\), delete \(\Z[j]\).
    \item For each subset \(\N_\iota \in \leakage{\X}\) where \(|\N| \leq
    \threshold\): 
    \begin{itemize}
        \item If \(\N_\iota \in \N_\adv\): for each \(\hat{\iota}
        \in \N_\iota\): let \(z_{\hat{\iota}} = \RO[x_{\hat{\iota}}]\), and
        construct \((c_{\hat{\iota}},s_{\hat{\iota}},t_{\hat{\iota}})\) as in
        Equation~\eqref{eq:simulate}.
        \item Else, sample \(K_\iota \sample \bin^\secpar\), and then for each \(\hat{\iota} \in \N_\iota\) compute:
        \begin{align}
            \begin{split}
                c_{\hat{\iota}} &\leftarrow \skeschemeshort.\enc(K_\iota,0);\\
                s_{\hat{\iota}} &\sample \FF_p;\\
                t_{\hat{\iota}} &\sample \bin^{\omega}.
            \end{split}
            \label{eq:simulate}
        \end{align}
    \end{itemize}
\end{itemize}

\begin{figure*}[t]
    \centering
    \large
    \scalebox{.68}{%
        \begin{tabular}{cccccc}\toprule
            Step & \voprfschemeshort.\RO queries & \(i \in \widehat{\cdv_\adv}\) & \((j \notin \widehat{\cdv_\adv}) \wedgespace (x_j \in \Y)\) & \((\iota \notin \widehat{\cdv_\adv}) \wedgespace (x_\iota \notin \Y)\) & Hop \\ [0.5ex] \midrule
            \(\game{0}\) & (\(x_j,f(\msk,x)\)) & \((c_i,s_i,t_i)\) & \((\skeschemeshort.\enc(K_j,x_j \| \aux_j),\adssschemeshort{\threshold}{n}.\share{r_{j,1};r_{j,2}},r_{j,3})\) & \((\skeschemeshort.\enc(K_\iota,x_\iota \| \aux_\iota),\adssschemeshort{\threshold}{n}.\share{r_{\iota,1};r_{\iota,2}},r_{\iota,3})\) & --- \\ 
            \(\game{1}\) & (\(x_j,\mathcolorbox{blue!20}{w \chkequal f(\msk,x)}\)) & \((c_i,s_i,t_i)\) & \((\skeschemeshort.\enc(K_j,x_j \| \aux_j),\adssschemeshort{\threshold}{n}.\share{r_{j,1};r_{j,2}},r_{j,3})\) & \((\skeschemeshort.\enc(K_\iota,x_\iota \| \aux_\iota),\adssschemeshort{\threshold}{n}.\share{\widetilde{r_{\iota,1}};\widetilde{r_{\iota,2}}},\widetilde{r_{\iota,3}})\) & ROM \\ 
            \(\game{2}\) & (\(x_j,\mathcolorbox{blue!20}{w \chkequal \F_{\voprfschemeshort}(x_j)}\)) & \((c_i,s_i,t_i)\) & \((\skeschemeshort.\enc(K_j,x_j \| \aux_j),\adssschemeshort{\threshold}{n}.\share{r_{j,1};r_{j,2}},r_{j,3})\) & \((\skeschemeshort.\enc(K_\iota,x_\iota \| \aux_\iota),\adssschemeshort{\threshold}{n}.\share{\widetilde{r_{\iota,1}};\widetilde{r_{\iota,2}}},\widetilde{r_{\iota,3}})\) & \voprf \\ 
            \(\game{3}\) & (\(x_j,\F_{\voprfschemeshort}(x_j)\)) & \((c_i,s_i,t_i)\) & \((\skeschemeshort.\enc(K_j,x_j \| \aux_j),\adssschemeshort{\threshold}{n}.\share{r_{j,1};r_{j,2}},r_{j,3})\) & \((\skeschemeshort.\enc(K_\iota,x_\iota \| \aux_\iota),\mathcolorbox{blue!20}{s_\iota \sample \FF_p},\widetilde{r_{\iota,3}})\) & \adss{\threshold}{n} \\ 
            \(\game{4}\) & (\(x_j,\F_{\voprfschemeshort}(x_j)\)) & \((c_i,s_i,t_i)\) & \((\skeschemeshort.\enc(K_j,x_j \| \aux_j),\adssschemeshort{\threshold}{n}.\share{r_{j,1};r_{j,2}},r_{j,3})\) & \((\mathcolorbox{blue!20}{\skeschemeshort.\enc(K_\iota,0 \ldots 0)},s_\iota \sample \FF_p,\widetilde{r_{\iota,3}})\) & \(\indcpa\) \\ \bottomrule
        \end{tabular}
    } 
    \caption{Game hops required to prove security of
    Theorem~\ref{thm:privacymalaggserverp2}. \game{0} corresponds to
    the real world execution of \(\P\), and \game{4} corresponds to
    the PPT simulator that interacts with the ideal functionality (Section~\ref{sec:security-model}).
    The third, fourth, and fifth columns correspond to the way that
    client messages are constructed in each game hop. The differences between each game hop are highlighted in blue.} 
    \label{table:thmp2aggserver}
\end{figure*}
In the following claims, we prove that the simulation is
indistinguishable to the adversary from the real protocol via a series
of game-hops. For a broad overview of the security proof,
see Figure~\ref{table:thmp2aggserver}.
\begin{claim}
    \(\game{0} \stat \game{1}\) in the random oracle model.
    \label{claim:thm1G1}
\end{claim}
\begin{proof}
    In \(\game{0}\), the execution is as in protocol \(\P\). In
    \(\game{1}\), all queries \((x,w)\) for the \(\voprfschemeshort.\RO\)
    are handled by first checking that \(w =
    f(\msk,x)\), which can be done using the master
    secret key sampled by the simulator. If the check passes, then the query is
    answered by either returning \(\voprfschemeshort.\RO[x]\) (if
    non-empty), or sampling a new value and assigning that to
    \(\voprfschemeshort.\RO[x]\), before returning it. If the check
    does not pass, then the query is answered by simply returning a
    random value. Note that the pseudorandomness property of
    \voprfschemeshort ensures that the two games are
    indistinguishable.
\end{proof}
\begin{claim}
    \(\game{1} \comp \game{2}\) by the security of \voprfschemeshort.
    \label{claim:thm1G2}
\end{claim}
\begin{proof}
    In \(\game{2}\) the simulator no longer has access to \(\msk\), and
    only has access to the ideal functionality
    \(\F_{\voprfschemeshort}\). Any blind evaluation query for \(x'\) is
    answered by sending the query to the corresponding interface of
    \(\F_{\voprfschemeshort}\), and returning the response to \adv. When
    \(\adv\) makes a query \((x,w)\) to \(\voprfschemeshort.\RO\) to finalize the VOPRF
    result, \(\sdv\) sends \((x)\) to the evaluation interface as a client input for
    \(\F_{\voprfschemeshort}\) to check if the queries are admissible in
    the same way as \game{1}. Note that the difference between \game{1}
    and \game{2} can be simulated by an adversary \(\bdv\) against
    \voprfschemeshort.
\end{proof}
\begin{claim}
    \(\game{2} \comp \game{3}\) by the share privacy of \(\adssschemeshort{\threshold}{n}\). 
    \label{claim:thm1G3}
\end{claim}
\begin{proof}
    In \game{3}, the simulator replaces all values \((s_i,t_i)\)
    sent by honest clients \(\CC_i\) where \(x_i \notin \Y\) and they
    belong to some \(\N_\iota \in \widehat{\N} \setminus \N_\adv\) (i.e., never queried
    to \RO by the adversary), with random values. The distinguishing advantage of the
    two games can be bounded by an adversary trying to break the
    privacy requirements of \(\adssschemeshort{\threshold}{n}\), since
    there are less than \(\threshold\) such shares. Moreover, the
    distribution of \(t_\iota\) is already random due to never having
    learnt the value from the output of \voprfschemeshort.
\end{proof}
\begin{claim}
    \(\game{3} \comp \game{4}\) by the \(\indcpa\) security of \skeschemeshort.
    \label{claim:thm1G4}
\end{claim}
\begin{proof}
    In \(\game{4}\), the only difference is that any message from an
    honest client \(\CC_i\) to the server \(\SS\) that encodes a measurement \(x_\iota\) that has not
    previously been queried to \(\voprfschemeshort.\RO\) are
    modified. In particular, these messages replace the encrypted
    ciphertext of the encoded message \((x_\iota \| \aux_i)\) with an
    encryption of all zeros (matching the length). The difference
    between these two games can be simulated by an adversary
    \(\bdv\) attempting to break the \(\indcpa\) security of
    \(\skeschemeshort\), since the encryption key is derived from
    randomness that \(\adv\) never witnessed. Note that the clients
    \(\CC_i\) that belong to this set can be learned from the output
    \Y and the output of the leakage function \(\leakage{\X}\).
\end{proof}

Note that the execution in \(\game{4}\) is identical to the view
described by the simulator above. Therefore, putting
Claim~\ref{claim:thm1G1}, Claim~\ref{claim:thm1G2},
Claim~\ref{claim:thm1G3}, and Claim~\ref{claim:thm1G4} together, we have
that the distinguishing advantage of the real-world execution and
the ideal world simulation is negligible and the proof of
Theorem~\ref{thm:privacymalaggserverp2} is complete.
\end{proof}

\begin{theorem}{\emph{(Malicious randomness
    server)}} The protocol \(\P\) is secure against any
    \(\adv\) that corrupts \OO and some subset \(\cdv_\adv \subset \cdv\) of
    all clients, assuming the security of \voprfschemeshort, and the \(\indcpa\) security of \skeschemeshort.
    \label{thm:privacymalrandserver}
\end{theorem}
\begin{proof}
    By the security of \voprfschemeshort, the simulator can simulate the view of \OO during the randomness phase of the protocol. During the aggregation phase, the server \OO also witnesses encrypted client messages that are destined for the aggregation server. Such encrypted messages can be simulated as encryptions of all zeroes in every case by the \(\indcpa\) security of \skeschemeshort. This simulates the entire view of \OO.
    
    Note that interactions with the ideal functionality can be made
    without submitting any adversarial client inputs, since these may be
    arbitrarily corrupted. Note that the simulator can thus only maintain
    correctness for \(\SS\) up to the output learnt purely from honest clients.
\end{proof}

\begin{theorem}{\emph{(Malicious clients)}} The protocol \(\P\) is
    secure against any adversary \(\adv\) corrupting
    some subset \(\cdv_\adv \subset \cdv\) of all clients, assuming the security of \voprfschemeshort, the
    \(\indcpa\) security of \(\skeschemeshort\), and the privacy of
    \(\adssschemeshort{\threshold}{n})\).
    \label{thm:robustness}
\end{theorem}
\begin{proof}
    This proof follows an almost identical set of transitions
    to the proof of Theorem~\ref{thm:privacymalaggserverp2}. Note that the
    simulator can simulate all messages in the same way, except that it only
    sends adversarial client messages to the ideal functionality
    that are \emph{well-constructed}. It can check whether messages are
    well-constructed by checking that the ciphertext \(c_i\) encrypts a
    value \(x_i\) that was received in the queries to
    \(\voprfschemeshort.\RO\), and using a correctly derived key. Note
    that this can be checked using the combination of inputs and outputs
    derived from \(\voprfschemeshort.\RO\). Similarly, the simulator can
    check whether \(s_i\) and \(t_i\) are consistent with the value of
    \(x_i\) that is encrypted. The simulator is then able to construct a set
    of messages using the output of \(\F_{\P}\) that
    provides the same correctness guarantees as in the real-world execution.
    The simulation for these messages is identical to the simulation in
    the proof of Theorem~\ref{thm:privacymalaggserverp2}.
\end{proof}

\point{Security for \toollite}
\label{app:starlite-sec}
The security of the \toollite protocol only holds when the client input distribution has sufficient min-entropy. Simulating security of unrevealed measurements against a malicious aggregation server is fairly trivial since the simulator can simply construct dummy-encodings, and rely upon the fact that the aggregation server is unable to guess which measurement is encoded with anything other than negligible probability. Otherwise, the simulation follows a similar model as the proof of Theorem~\ref{thm:privacymalaggserverp2}. In the case of malicious clients, security is ensured by modelling the randomness derivation process as interacting with a random oracle model, and the proof follows identically to Theorem~\ref{thm:robustness}.

\end{document}